\documentclass[
letterpaper,
showpacs,showkeys,
superscriptaddress,
twocolumn,
nofootinbib,aps,pra,10pt]{revtex4-1}

\pdfoutput=1

\usepackage{dsfont}
\usepackage{amsthm,amssymb,amsmath,amsfonts,amscd}
\usepackage{braket}
\usepackage{enumerate}
\usepackage{dsfont}
\usepackage{mathtools}
\usepackage{natbib}
\usepackage[colorlinks=true, linkcolor=black, citecolor=blue, urlcolor=blue, plainpages=false, pdfpagelabels, backref=none, hypertexnames=false, unicode]{hyperref}
\usepackage{color}
\usepackage[utf8x]{inputenc}
\usepackage{lmodern}
\usepackage{url}

\theoremstyle{plain}
\newtheorem{theorem}{Theorem}
\newtheorem{lemma}[theorem]{Lemma}
\newtheorem{corollary}[theorem]{Corollary}

\renewcommand{\H}{H}
\def\Tr{\mathrm{Tr}}
\def\RR{\mathbb{R}}
\def\CC{\mathbb{C}}
\DeclareMathOperator{\dom}{dom}
\providecommand{\abs}[1]{\left\lvert#1\right\rvert}
\providecommand{\norm}[1]{\left\lVert#1\right\rVert}

\begin{document}

\title{On convex optimization problems in quantum information theory}

\author{Mark W. Girard}
\email{mwgirard (at) ucalgary.ca}
\affiliation{Institute for Quantum Science and Technology, University of Calgary}
\affiliation{Department of Mathematics and Statistics, University of Calgary, 2500 University Dr NW, Calgary, Alberta T2N 1N4, Canada}
\author{Gilad Gour}
\affiliation{Institute for Quantum Science and Technology, University of Calgary}
\affiliation{Department of Mathematics and Statistics, University of Calgary, 2500 University Dr NW, Calgary, Alberta T2N 1N4, Canada}
\author{Shmuel Friedland}
\affiliation{Department of Mathematics, Statistics and Computer Science, University of Illinois at Chicago, 851 S. Morgan St, Chicago, Illinois 60607, USA}

\date{\today}

\begin{abstract}
Convex optimization problems arise naturally in quantum information theory, often in terms of minimizing a convex function over a convex subset of the space of hermitian matrices. In most cases, finding exact solutions to these problems is usually impossible. As inspired by earlier investigations into the relative entropy of entanglement [Phys. Rev. A \href{http://dx.doi.org/10.1103/PhysRevA.78.032310}{\textbf{78} 032310} (2008)], we introduce a general method to solve the converse problem rather than find explicit solutions. That is, given a matrix in a convex set, we determine a family of convex functions that are minimized at this point.  This method allows us find explicit formulae for the relative entropy of entanglement and the Rains bound, two well-known upper bounds on the distillable entanglement, and yields interesting information about these quantities, such as the fact that they coincide in the case where at least one subsystem of a multipartite state is a qubit. 
\end{abstract}


\maketitle

\section{Introduction}
\label{sec:intro}
Convexity naturally arises in many places in quantum information theory; the sets of possible preparations, processes and measurements for quantum systems are all convex sets. Indeed, many important quantities in quantum information are defined in terms of a convex optimization problem.  In particular, entanglement is an important resource~\cite{Plenio2007,Horodecki2009} and quantifying entanglement is a problem that is often cast in terms of convex optimization problems.

Since the set of separable or unentangled states is convex, a measure of entanglement may be defined for entangled states outside of this set, given a suitable `distance' measure, as the minimum distance to a state inside. If $\mathcal{D}$ is the set of all unentangled states, a measure of entanglement for a state~$\rho$ can be given by
\begin{equation*}
 E(\rho)=\min_{\sigma\in\mathcal{D}}D(\rho\| \sigma),
\end{equation*}
where $D(\rho\| \sigma)$ is a suitable distance (though not necessarily a metric) between two states~$\rho$ and~$\sigma$~\cite{Vedral1998}. 

Perhaps the most well known of these quantities is the relative entropy of entanglement $E_\mathcal{D}$, in which the choice of distance function is given by the relative entropy $D(\rho\| \sigma)=S(\rho\| \sigma)$, defined as
\begin{equation*}
\label{eq:REE}
 S(\rho\| \sigma)=\Tr[\rho(\log\rho-\log\sigma)].
\end{equation*}
However, even in the simplest case of a system of two qubits, determining whether or not a closed formula for the relative entropy of entanglement exists is still an open problem \mbox{\cite[prob. 8]{Krueger2005}}.

Many other quantities in quantum information can be considered in terms of convex optimization problems. Defining $\H_n$ as the space of $n\times n$ hermitian matrices, these problems are usually given in terms of a convex function $f:\H_n\rightarrow[0,+\infty]$ and some convex subset~$\mathcal{C}\subset \H_n$.  Then we can ask the question, ``when does a matrix~$\sigma^\star\in\mathcal{C}$ minimize~$f$?'' That is, when does a matrix~$\sigma^\star\in\mathcal{C}$ satisfy
\begin{equation*}
f(\sigma^\star)=\min_{\sigma\in\mathcal{C}}f(\sigma)
\end{equation*}
assuming that $f(\sigma)$ is finite for at least one $\sigma\in\mathcal{C}$. Since~$f$ is a convex function, it is sufficient to show that~$\sigma^\star$ is a semi-local minimum of $f$ (see eq. \eqref{eq:maintheorem-criterion} in Theorem~\ref{thm:maintheorem}). That is, all of the directional derivatives of $f$ at $\sigma^\star$ are nonnegative.
Since the directional derivative is usually a linear function, denoted by $D_{f,\sigma^\star}:\H_n\to \RR$, 
the condition in eq.~\eqref{eq:maintheorem-criterion} reduces to the fact that $D_{f,\sigma^\star}$ defines a supporting functional of $\mathcal{C}$ at $\sigma^\star$.

In general, finding a closed analytic formula for an optimal~$\sigma^\star$ is difficult if not impossible.  Although, from the computational point of view, the complexity of finding a good approximation to $\sigma^\star$ is relatively easy, i.e. polynomial, in terms of computation of the function $f$ and the membership in $\mathcal{C}$.  Such numerical optimization for the relative entropy of entanglement, for example, has been studied in~\cite{Zinchenko2010}.

However, rather than trying to directly solve these optimization problems, in this paper we instead discuss methods to solve the converse problem.  That is, given a matrix~$\sigma^\star\in\mathcal{C}$, we instead ask, ``which functions $f$ achieve their minimum over $\mathcal{C}$ at $\sigma^\star$?'' Although this may seem trivial at first, these kinds of results can yield meaningful statements about finding closed formula for certain quantities in quantum information~\cite{Friedland2011}. 

Recent work has been done~\cite{Ishizaka2003,Miranowicz2008,Friedland2011} that employs similar methods to determine an explicit expression for the relative entropy of entanglement for certain states.  Given a state $\sigma^\star\in\mathcal{D}$, one can find all of the entangled states~$\rho$ for which $\sigma^\star$ is the \emph{closest} separable state, thus minimizing the relative entropy of entanglement.  We extend these results to find a explicit expression for the Rains bound~\cite{Rains1999}, a quantity related to the relative entropy of entanglement, and show how these results can be generalized to other functions of interest in quantum information theory. 

The remainder of this paper is outlined as follows. In section \ref{sec:maintheorem}, we present the necessary and sufficient conditions needed for a matrix~$\sigma^\star\in\mathcal{C}$ to minimize a convex function $f$ over an arbitrary convex subset $\mathcal{C}\subset\H_n$. The examples of this analysis applied to the relative entropy of entanglement and the Rains bound are presented in sections \ref{sec:REE} and \ref{sec:Rains} respectively. In section~\ref{sec:REEcompareRains}, these results are then used to prove some interesting facts about these two quantites, such as the fact that the Rains bound and the relative entropy of entanglement coincide for states in quantum systems in which at least one subsystem is a qubit. Further applications to other convex functions are contained in section \ref{sec:OtherApplications} while section \ref{sec:conclusion} concludes.

\section{Necessary and sufficient conditions for minimizing a convex function}
\label{sec:maintheorem}

We first recall some basic definitions. Let $M_n$ be the space of $n\times n$ matrices and $ \H_n$ the subset of hermitian matrices. Given an interval $I\in\mathbb{R}$, denote $H_n(I)$ as the subset of hermitian matrices whose eigenvalues are contained in $I$, where $I$ may be open, closed or half-open. We also define the subsets 
\[ \H_{n,+,1}\subset \H_{n,+}\subset \H_{n},\]
where  $ \H_{n,+}$ is the cone of positive hermitian matrices and~$\H_{n,+,1}$ consists of the positive hermitian matrices with unit trace. Note that $ \H_{n,+,1}$ coincides with the space of density matrices acting on an $n\times n$ quantum system, which may be composed of subsystems of dimension $n_1\times\cdots\times n_k=n$. Furthermore, let $H_{n,++}=H_n(0,+\infty)$ denote the set of hermitian matrices whose eigenvalues are strictly positive.
For $A,B\in\H_{n}$, we denote by $A\leq B$ when $B-A\in\H_{n,+}$ and $A< B$ when $B-A\in\H_{n,++}$.  

With the Hilbert-Schmidt inner product given by 
\[
 \langle A,B\rangle=\Tr[A^\dagger B],
\]
 the space $M_n$ becomes a Hilbert space and $ \H_n$ becomes a real Hilbert space. A linear superoperator $\Lambda:M_n\rightarrow M_n$ is said to be \emph{self-adjoint} if it is self-adjoint with respect to the Hilbert-Schmidt inner product, i.e.
 \begin{equation*}
  \Tr[\Lambda(A)^\dagger B]= \Tr[A^\dagger \Lambda(B)] \hspace{2em}\textrm{ for all }A,B\in M_n.
 \end{equation*}
  
While we will generally be using $\H_n$ as our Hilbert space, since this is the most interesting one in quantum information, the main theorem of this paper also applies to any Hilbert space $\mathcal{H}$.

Let $\mathcal{C}\subset \mathcal{H}$ be a convex set. We recall that a function 
\mbox{$f:\mathcal{C}\rightarrow \RR$} is said to be \emph{convex}~if
\begin{equation}\label{eq:maintheorem_convex}
 f\big((1-t)A+tB\big)\leq (1-t)f(A)+tf(B)
\end{equation}
for all $A,B\in \mathcal{C}$ and $t\in[0,1]$, and $f$ is \emph{concave} if $-f$ is convex. Furthermore, a convex function \mbox{$f:\mathcal{C}\rightarrow \mathcal{H}$} is said to be \emph{operator convex} if  the above relation holds as a matrix inequality.  We call $f$ strictly convex if \[ f\big((1-t)A+tB\big)< (1-t)f(A)+tf(B)\] for all $t\in (0,1)$ and $A\ne B$.

Instead of only considering functions $f:\mathcal{C}\rightarrow\RR$, it is convenient to allow the range of $f$ to be the extended real line $\RR^{+\infty}$, defined as~\mbox{$\RR^{+\infty}=(-\infty,+\infty]$}. We can define the \emph{domain}~of~$f$ as $\dom f=f^{-1}(\RR)$, i.e. the set of $\sigma\in\mathcal{C}$ such that $f(\sigma)$ is finite. A function $f:\mathcal{C}\rightarrow\RR^{+\infty}$ is said to be \emph{proper} if~$\dom f\neq\emptyset$, and $f$ is said to be convex if it is convex satisfying eq.~\eqref{eq:maintheorem_convex} on its domain. Furthermore, if $f$ is convex then $\dom f\subset\mathcal{C}$ is a convex subset. Finally, a function $f:\mathcal{C}\rightarrow[-\infty,+\infty)$ is \emph{concave} if $-f$ is convex on its domain (see for example~\cite[ch. 2]{Borwein2010}).

\subsection{Necessary and sufficent conditions for optimization}

Given a convex and compact subset $\mathcal{C}\subset\mathcal{H}$ and a convex function $f:\mathcal{C}\rightarrow \RR^{+\infty}$,  solving for an element~$\sigma^\star$ in~$\mathcal{C}$ that minimizes~$f$ is usually a daunting task.  Yet, in the following theorem, we state a necessary and sufficient condition for $\sigma^\star$ to minimize $f$, one which involves convex combinations of the form $(1-t)\sigma^\star+t\sigma$ where~$\sigma\in\mathcal{C}$.  

Here, we make use of the \emph{directional derivative} of $f$ at a point $A\in\mathcal{C}$. For $A$ in the domain of $f$ and $B\in\mathcal{H}$, this is defined by
\begin{equation}\label{eq:maintheorem_directionalderivative}
 f'(A;B):=\lim_{\,\,\,t\rightarrow 0^+}\frac{f(A+tB)-f(A)}{t}
\end{equation} 
if the limit exists. Since $f$ is convex, this limit always either exists or is infinite. For example, if $A$ is on the boundary of the domain of $f$, it is possible that $f'(A;B)=\pm\infty$.
(For example if $\mathcal{C}=H_{n,+}$ and $f(A)=\Tr(-\log A) $, then for $A,B\in H_{n,+}$ with $A\not>0$ and $B>0$ we have that $f'(A;B)=\infty$.)  If $A$ is in the interior of the domain, then $f'(A;B)$ is finite for all $B\in\mathcal{H}$.
\begin{theorem}\label{thm:maintheorem}
 Let $\mathcal{H}$ be a Hilbert space, $\mathcal{C}\subset \mathcal{H}$ a convex compact
 subset and $f:\mathcal{H}\rightarrow \RR^{+\infty}$ a convex
 function. Then an element $\sigma^\star\in\mathcal{C}$ minimizes~$f$ over~$\mathcal{C}$, i.e. \mbox{$\displaystyle\min_{\sigma\in\mathcal{C}}f(\sigma)=f(\sigma^\star)$}, if and only if for all $\sigma\in\mathcal{C}$
 \begin{equation}\label{eq:maintheorem-criterion}
  f'(\sigma^\star;\sigma-\sigma^\star)\geq 0.
 \end{equation}
\end{theorem}
\noindent (See for example~\cite[ch 2.1]{Borwein2006} and~\cite[p. 147]{Borwein2010}).
\begin{proof}
 Let $\sigma^\star\in\mathcal{C}$. If $f(\sigma^\star)\leq f(\sigma)$ for all $\sigma\in\mathcal{C}$, then clearly $f$ cannot decrease under a small perturbation away from~$\sigma^\star$. So eq.~\eqref{eq:maintheorem-criterion} must hold for all $\sigma\in\mathcal{C}$.  Now suppose that eq.~\eqref{eq:maintheorem-criterion} holds for all $\sigma\in\mathcal{C}$. For a fixed~$\sigma$, consider the function
 \begin{equation*}
  h(t)= f\big((1-t)\sigma^\star+t\sigma\big)
 \end{equation*}
 on the interval $t\in[0,1]$. Note that $h$ is convex by convexity of $f$. Since $h'(0)\geq 0$, this implies that $h$ must be non-decreasing on $[0,1]$. In particular, this means that $h(0)\leq h(1)$ and so $f(\sigma^\star)\leq f(\sigma)$. Since this holds for all $\sigma\in\mathcal{C}$, the minimum of $f$ over~$\mathcal{C}$ is achieved at~$\sigma^\star$ as desired. 
\end{proof}

The main idea of this paper is to turn the criterion in eq.~\eqref{eq:maintheorem-criterion} into one that is more useful so that we may more easily determine if a given~$\sigma^\star$ is optimal.  In case the value of the directional derivative is linear in the choice of~$\sigma$, the criterion $f'(\sigma^\star;\sigma-\sigma^\star)\geq 0$ can be recast in terms of a \emph{supporting functional} for the convex set~$\mathcal{C}$. That is, a linear functional $\Phi:\mathcal{H}\rightarrow\RR$ such that $\Phi(\sigma)\leq c$ for all~$\sigma\in\mathcal{C}$, and $c=\Phi(\sigma^\star)$ for some $\sigma^\star\in\mathcal{C}$. 
Then the set $\left\{\sigma\in\mathcal{H}\,|\,\Phi(\sigma)=c\right\}$ is a \emph{supporting hyperplane} tangent to~$\mathcal{C}$ at the point $\sigma^\star\in\mathcal{C}$.

In a finite-dimensional Hilbert space, any linear functional can be written as~$\Phi(\sigma)=\langle\phi,\sigma\rangle$ for some~\mbox{$\phi\in\mathcal{H}$}. Since the Hilbert space considered here is $H_n$, the linear functionals are of the form $\Phi(\sigma)=\Tr[\phi\sigma]$ for some matrix $\phi\in\H_n$. Conversely, every matrix $\phi\in\H_n$ defines a linear functional $\Phi:\H_n\rightarrow\RR$ by $\Phi(\sigma)=\Tr[\phi\sigma]$.

\subsection{Linear differential operator \texorpdfstring{$D_{g,}$}{Dg}}

Going back to applications in quantum information, we first examine functions $f:\H_n(a,b)\rightarrow \RR$ of the form
\begin{equation}
 f(\sigma)=-\Tr[\rho g(\sigma)],
\end{equation}
where $\rho\in \H_{n,+}$ is a matrix and $g:(a,b)\rightarrow \mathbb{R}$ is an analytic function that we can extend to matrices in $\H_n(a,b)$.  Since~$g$~is analytic, we may easily take the necessary derivatives to investigate the criterion in Theorem~\ref{thm:maintheorem} (see for example~\cite[ch. 6]{Horn1991}). We then extend this analysis to extended-real value functions $f:H_n\rightarrow\RR^{+\infty}$ by considering carefully the cases when $\sigma\not\in\H_n(a,b)$, in which case $f(\sigma)=+\infty$ for most (but not necessarily all) $\sigma\not\in\H_{n}(a,b)$.

Let $g:(a,b)\rightarrow \mathbb{R}$ be an analytic function and $\Omega\subset\CC$ an open set containing $(a,b)$ such that $g$  can be extended to~$g:\Omega\rightarrow\CC$, a function that is analytic on~$\Omega$. If $A\in M_n$ is an $n\times n$ matrix whose eigenvalues are contained in $\Omega$, we may write (see~\cite[ch. 6.1]{Horn1991})
\[
 g(A)=\frac{1}{2\pi i}\oint_\gamma g(s)[s\mathds{1}-A]^{-1}ds,
\]
where $\gamma$ is any simple closed recitifiable curve in $\Omega$ that encloses the eigenvalues of $A$.
For a continuously differentiable family $A(t)\in \H_n$ of hermitian matrices, and $t_0$ such that the eigenvalues of $A(t_0)$ is contained in $\Omega$, we have 
\begin{multline*}
 2\pi i \left.\frac{d}{dt}g(A(t))\right|_{t=t_0}=\\
 \begin{array}{l}
 \hspace{3mm}=\displaystyle\oint_\gamma g(s)\frac{d}{dt}\left.\left([s\mathds{1}-A(t)]^{-1}\right)\right|_{t=t_0}ds\\
 \hspace{3mm}=\displaystyle\oint_\gamma g(s)[s\mathds{1}-A(t_0)]^{-1}A'(t_0)[s\mathds{1}-A(t_0)]^{-1}ds,
 \end{array}
\end{multline*}
where $A'(t_0)=\left.\frac{d}{dt}A(t)\right|_{t=t_0}$, 
and $\gamma$ is any simple, closed rectifiable curve in $\Omega$ that encloses all the eigenvalues of~$A(t_0)$. 
Since $A(t)$ is hermitian, we may write it in terms of its spectral decomposition (in particular at $t=t_0$) as
\[
A(t_0)=\sum_{i=1}^n a_i\ket{\psi_i}\bra{\psi_i}
\]
where $\ket{\psi_i}$ are the orthonormal eigenvectors of $A(t_0)$ and~$a_i$ the corresponding eigenvalues. Then the matrix elements of $\left.\frac{d}{dt}g(A(t))\right|_{t=t_0}$ in the eigenbasis of $A(t_0)$ are
\begin{multline*}
 \bra{\psi_i}\left(\left.\frac{d}{dt}g(A(t))\right|_{t=t_0}\right)\ket{\psi_j}=
 \\ \begin{array}{l}
     =\bra{\psi_i}A'(t_0)\ket{\psi_j}\displaystyle{\frac{1}{2\pi i}\oint_{\gamma}\frac{g(s)}{(s-a_i)(s-a_j)}ds}
 \\ =\bra{\psi_i}A'(t_0)\ket{\psi_j}\cdot
 \left\{
	  \begin{array}{ll}
          \frac{g(a_i)-g(a_j)}{a_i-a_j} & a_i\neq a_j\\
          g'(a_i)& a_i=a_j.
         \end{array}\right.
\end{array}
\end{multline*}
For a function $g:(a,b)\rightarrow\RR$ and an hermitian matrix~$A$ whose eigenvalues $\{a_1,\dots,a_n\}$ are contained in $(a,b)$, we define the matrix of the so-called \emph{divided differences}~(see~\cite[p. 123]{Bhatia1997})~as
\[
 \left[ T_{g,A}\right]_{ij}=\left\{\begin{array}{ll}
          \frac{g(a_i)-g(a_j)}{a_i-a_j} & a_i\neq a_j\\
          g'(a_i)& a_i=a_j.
         \end{array}
\right. 
\]
In the eigenbasis of $A(t_0)$, i.e. assuming that $A(t_0)$ diagonal, we may write 
\[
 \left.\frac{d}{dt}g(A(t))\right|_{t=t_0} =  T_{g,A}\circ A'(t_0) ,
\]
where $\circ$ represents the Hadamard (entrywise) product of matrices. The entrywise product of $ T_{g,A}$ with a matrix is a linear operator on the space of matrices, and for an arbitrary matrix $B\in M_n$ we write
\begin{equation}\label{eq:FrechetDef}
D_{g,A}(B)=  T_{g,A} \circ B,
\end{equation}
where $D_{g,A}$ is a linear operator on the space of matrices such that $\left.\frac{d}{dt}g(A(t))\right|_{t=t_0}=D_{g,A}(A'(0))$. This linear operator $D_{g,A}:M_n\rightarrow M_n$ is called the \emph{Fr\'echet derivative} of $g$ at $A$ (see for example~\cite[ch. X.4]{Bhatia1997}).  A function $g:H_n(a,b)\rightarrow H_n$ is said to be \emph{Fr\'echet differentiable} at a point $A$ when it's directional derivative
\[
 g'(A;B):=\lim_{\,\,\,t\rightarrow 0^+}\frac{g(A+tB)-g(A)}{t}
\]
is linear in $B$ and coincides with the Fr\'echet derivative, i.e. $D_{g,A}(B)=g'(A;B)$.

Furthermore, as long as $g'(a_i)\neq 0$ for all eigenvalues~$a_i$ of $A$, the linear operator $D_{g,A}$ is invertible, since we may define a matrix $ S_{g,A}$ as the  element-wise inverse of $ T_{g,A}$, namely 
\[
 \left[ S_{g,A}\right]_{ij}=\left\{\begin{array}{ll}
          \frac{a_i-a_j}{g(a_i)-g(a_j)} & a_i\neq a_j\\
          \frac{1}{g'(a_i)}& a_i=a_j,
         \end{array}
\right. 
\]
such that $ T_{g,A}\circ S_{g,A}\circ B = S_{g,A}\circ T_{g,A}\circ B = B$ for all matrices $B$. Then the inverse of the linear operator $D_{g,A}$ is then given by $D_{g,A}^{-1}(B)= S_{g,A}\circ B$, such that $D_{g,A}^{-1}(D_{g,A}(B)) = D_{g,A}(D_{g,A}^{-1}(B))=B$ for all matrices~$B$.

In the following section, we consider the extended-real valued function $f:\H_n\rightarrow \RR^{+\infty}$, in which case it is important to also extend the definition of $D_{g,A}$ to matrices $A$ whose eigenvalues are not in $(a,b)$.  The definition of $D_{g,A}$ is extended in the following manner. For $A\not\in\H(a,b)$, define the matrix $T_{g,A}$ as above on the eigenspaces of $A$ with corresponding eigenvalues in $(a,b)$, but to be zero otherwise $A$. Thus, in the eigenbasis of $A$, the matrix elements of $T_{g,A}$ are
\[
 \left[ T_{g,A}\right]_{ij}=\left\{\begin{array}{ll}
          \frac{g(a_i)-g(a_j)}{a_i-a_j} & a_i\neq a_j,\,a_i,a_j\in(a,b)\\
          g'(a_i)&a_i=a_j\in(a,b)\\
          0 & a_i\textrm{ or }a_j\not\in(a,b),
         \end{array}
\right. 
\]
such that $D_{g,A}(B)=T_{g,A}\circ B$. If $g'(a_i)\neq 0$ for all eigenvalues $a_i$ of $A$ in the interval $(a,b)$, then we can define the \emph{Moore-Penrose} inverse (or \emph{pseudo}-inverse) of $D_{g,A}$. In the eigenbasis of $A$, this is given by \mbox{$D_{g,A}^\ddagger(B)=S_{g,A}\circ B$} where $S_{g,A}$ is the matrix with matrix elements given in the eigenbasis of $A$ as
\[
[S_{g,A}]_{ij}=\left\{\begin{array}{ll}
          \frac{a_i-a_j}{g(a_i)-g(a_j)} & a_i\neq a_j,\,a_i,a_j\in(a,b)\\
          \frac{1}{g'(a_i)}& a_i=a_j\in(a,b)\\
          0 & a_i \textrm{ or }a_j\not\in(a,b),
         \end{array}
\right. 
\]
such that $D_{g,A}(D_{g,A}^\ddagger(B))=D_{g,A}^\ddagger(D_{g,A}(B)) =P_A B P_A$ for all matrices $B$, where $P_A$ is the projection matrix that projects onto the eigenspaces of $A$ whose corresponding eigenvalues are in $(a,b)$.  If $A\in\H_n(a,b)$, then $D_{g,A}^\ddagger$ coincides with $D_{g,A}^{-1}$. 

Finally, we note that the linear differential operator~$D_{g,A}$ is self-adjoint with respect to the trace inner product. Indeed, the matrix $ T_{g,A}$ is hermitian since $A$ is hermitian and  
\begin{align*}
 \Tr[C^\dagger D_{g,A}(B)] &= \Tr[C^\dagger ( T_{g,A}\circ B)] \\&= \Tr[( T_{g,A}\circ C)^\dagger B] =  \Tr[(D_{g,A}(C))^\dagger B],
\end{align*}
for all matrices $B$ and $C$.


\subsection{Functions of form \texorpdfstring{$f_\rho(\sigma)=-\Tr[\rho g(\sigma)]$}{f(s)=-Tr[pg(s)]}}

Given a matrix $\rho\in \H_{n,+}$ we can now consider functions of the form $f_\rho(\sigma)=-\Tr[\rho g(\sigma)]$, which is convex as long as the function $g:(a,b)\rightarrow\mathbb{R}$ is concave.  We can extend $f_\rho$ to an extended-real valued convex function $f_\rho:H_n\rightarrow\RR^{+\infty}$ in the following manner. If $\rho\in\H_{n,++}$ is strictly positive, then $f_\rho(\sigma)=+\infty$ whenever~$\sigma\not\in\H_{n}(a,b)$.  Otherwise, if $\rho\in\H_{n,+}$ has at least one zero eigenvalue, in the eigenbasis of $\rho$ we can write the matrices $\rho$ and $\sigma$ in block form as
\begin{equation}
 \label{eq:maintheorem_blockform}
 \rho=\begin{pmatrix}
       \tilde{\rho}& 0\\
       0& 0
      \end{pmatrix}
\hspace{5mm}
\textrm{ and }
\hspace{5mm}
 \sigma=\begin{pmatrix}
       \tilde{\sigma}& \tilde{\sigma}_{12}\\
       \tilde{\sigma}_{21}& \tilde{\sigma}_{22}
      \end{pmatrix},
\end{equation}
such that $\tilde{\rho}\in\H_{\tilde{n},++}$, where $\tilde{n}$ is the dimension of the support of $\rho$. Then, if all of the eigenvalues of $\tilde{\sigma}$ are in~$(a,b)$, the matrix-valued function $g(\tilde{\sigma})$ is well-defined, and we can define \[f_\rho(\sigma)=f_{\tilde{\rho}}(\tilde{\sigma})=-\Tr[\tilde{\rho}g(\tilde{\sigma})].\] Otherwise,  define $f_\rho(\sigma)=+\infty$ if $\tilde{\sigma}\not\in\H_{\tilde{n}}(a,b)$.

The standard formula for the relative entropy of~$\rho$ with respect to~$\sigma$ is recovered by choosing the function \mbox{$g:(0,\infty)\rightarrow\RR$} to be \mbox{$g(x)=-S(\rho)+\log(x)$}, where $S(\rho)=\Tr[\rho\log\rho]$ is a constant for a fixed $\rho$. Namely,
\[
 f_\rho(\sigma)=\Tr[\rho(\log\rho-\log\sigma)]=S(\rho\| \sigma).
\]
In particular, if $\sigma>0$ then $f_\rho(\sigma)=S(\rho\| \sigma)$ is finite for all $\rho$, and $S(\rho\| \sigma)=+\infty$ if $\sigma$ is zero on the support of $\rho$, i.e. if $\bra{\psi}\sigma\ket{\psi}=0$ for some $\ket{\psi}$ such that $\bra{\psi}\rho\ket{\psi}>0$.

We now show how to find the derivatives of the function $f_\rho$. If $\sigma\in\H_n(a,b)$ then the directional derivative $f_\rho'(\sigma;\cdot)=D_{f_\rho,\sigma}(\cdot)$ is a well-defined linear functional and
\[
 D_{f_\rho,\sigma}(\tau)=-\Tr[\rho D_{g,\sigma}(\tau)]
\]
where $D_{g,\sigma}$ is the linear operator \mbox{$D_{g,\sigma}:\H_n\rightarrow\H_n$} defined previously. If $\sigma$ is in $\dom f_\rho$ but $\sigma\not\in\H_n(a,b)$, then we can compute the directional derivatives in the following manner. 

Given a concave analytic function $g:(a,b)\rightarrow\RR$ and a matrix $\rho\in\H_{n,+}$, consider the convex extended function $f_\rho:H_n\rightarrow\RR^{+\infty}$  given by \mbox{$f_\rho(\sigma)=-\Tr[\rho g(\sigma)]$}.  Let \mbox{$\sigma\in\dom f_\rho$} such that $\tilde{\sigma}$, as defined in equation~\eqref{eq:maintheorem_blockform}, is in $\H_{\tilde{n}}(a,b)$, and thus $f_\rho(\sigma)=f_{\tilde \rho}(\tilde{\sigma})$. Let $\tau\in\H_n$ and define $\tilde{\tau}\in\H_{\tilde{n}}$ as the block of $\tau$ on the support of $\rho$, analogous to equation~\eqref{eq:maintheorem_blockform}, i.e.
\[
  \tau=\begin{pmatrix}
       \tilde{\tau}& \tilde{\tau}_{12}\\
       \tilde{\tau}_{21}& \tilde{\tau}_{22}
      \end{pmatrix}.
\]
Since $\H_n(a,b)$ is open, there exists an $\epsilon>0$ small enough such that $\tilde{\sigma}+t\tilde{\tau}\in\H_{\tilde{n}}(a,b)$ for all $t\in[0,\epsilon)$, and thus $f_\rho(\sigma+t\tau)$ is finite for $t\in[0,\epsilon)$. Therefore the directional derivative $f_\rho'(\sigma;\tau)$ exists for all $\tau\in\H_n$ and is given by
\begin{equation}
\label{eq:maintheorem_DirectionalDerivModified}
 f_\rho'(\sigma;\tau) = -\Tr[\rho D_{g,\sigma}(\tau)].
\end{equation}
Indeed, consider $\rho$, $\sigma$ and $\tau$ in block form as in equation~\eqref{eq:maintheorem_blockform} such that $\tilde{\rho}>0$. Then $f_\rho(\sigma+t\tau)=f_{\tilde{\rho}}(\tilde{\sigma}+t\tilde{\tau})$ for all $t\in[0,\epsilon)$ and
\begin{align*}
f_\rho'(\sigma;\tau)
&=\lim_{t\rightarrow 0^+}\frac{f_{\tilde{\rho}}(\tilde{\sigma}+t\tilde{\tau})-f_{\tilde{\rho}}(\tilde{\sigma})}{t}\\
&=-\Tr\left[\tilde{\rho}\lim_{t\rightarrow 0^+}\frac{g(\tilde{\sigma}+t\tilde{\tau})-g(\tilde{\sigma})}{t}\right]\\
&=-\Tr[\tilde{\rho}g'(\tilde{\sigma};\tilde{\tau})]\\
&= -\Tr[\tilde{\rho}D_{g,\tilde{\sigma}}(\tilde{\tau})].
\end{align*}
Note that $\Tr[\tilde{\rho}D_{g,\tilde{\sigma}(\tilde{\tau})}]=\Tr[\rho D_{g,\sigma}(\tau)]$, so this simplifies to the form in equation~\eqref{eq:maintheorem_DirectionalDerivModified}.

We can now restate the criterion in Theorem \ref{thm:maintheorem} for functions of the form $f_\rho(\sigma)=-\Tr[\rho g(\sigma)]$ in terms of supporting functionals.
\begin{theorem}
\label{thm:maintheorem_secondary}
 Let $\mathcal{C}\subset  \H_{n}$ be a convex compact subset, \mbox{$\rho\in \H_{n,+}$},  and $g:(a,b)\rightarrow\RR$ be a concave analytic function. 
As above, consider the convex extended-real valued function \mbox{$f_\rho:\H_n\rightarrow\RR^{+\infty}$} given by \mbox{$f_\rho(\sigma)=-\Tr[\rho g(\sigma)]$}. Then a matrix $\sigma^\star\in\mathcal{C}$ minimizes~$f_\rho$ over~$\mathcal{C}$ if and only if 
 \begin{equation}\label{eq:FirstSHP}
  \Tr[D_{g,\sigma^\star}(\rho)\sigma]\leq\Tr[D_{g,\sigma^\star}(\rho)\sigma^\star]
 \end{equation}
for all $\sigma\in\mathcal{C}$. Thus the matrix $D_{g,\sigma^\star}(\rho)\in\H_n$ defines a supporting functional of $\mathcal{C}$ at $\sigma^\star$.
\end{theorem}
\begin{proof}
From Theorem \ref{thm:maintheorem}, we have that~$\sigma^\star$ minimizes~$f_\rho$ over~$\mathcal{C}$ if and only if $f_\rho'(\sigma^\star;\sigma-\sigma^\star)\geq0$ for all $\sigma\in\mathcal{C}$. Note that $\sigma^\star$ must be in the domain of~$f_\rho$. From the analysis above, we see that the directional derivative is
\[
 f_\rho'(\sigma^\star;\sigma-\sigma^\star)=-\Tr[\rho D_{g,\sigma^\star}(\sigma-\sigma^\star)].
\]
Since $D_{g,\sigma^\star}$ as a linear operator is self-adjoint with respect to the trace inner product, this yields the inequality in~\eqref{eq:FirstSHP} for all $\sigma\in\mathcal{C}$ if and only if $\sigma^\star$ minimizes $f_\rho$ over $\mathcal{C}$, as desired.
\end{proof}

Note that the matrix $D_{g,\sigma^\star}(\rho)\in \H_n$ defines a linear functional~$\Phi:\H_n\rightarrow\RR$ given by $\Phi(\sigma)=\Tr[D_{g,\sigma^\star}(\rho)\sigma]$ such that the criterion
\[
\Phi(\sigma)\leq\Phi(\sigma^\star) \,\,\textrm{ for all }   \sigma\in\mathcal{C}                                                                                                                                                                                                              \]
defines a supporting functional of~$\mathcal{C}$ at the point~\mbox{$\sigma^\star\in\mathcal{C}$}.  
This acts as ``witness'' for~$\mathcal{C}$ in the sense that, for \mbox{$\phi=D_{g,\sigma^\star}(\rho)$} and a constant $c=\Tr[\phi\sigma^\star]$, if $\sigma\in\H_{n}$ is a matrix such that $\Tr[\phi\sigma]>c$ then $\sigma\not\in\mathcal{C}$. Furthermore, if~$\sigma^\star$ is on the boundary of~$\mathcal{C}$, the hyperplane defined by the set of all matrices $\sigma\in\H_{n}$ such that $\Tr[\phi\sigma]=c$ is tangent to~$\mathcal{C}$ at the point~$\sigma^\star$. This criterion is only useful in characterizing the function~$f_\rho$, however, if the linear functional~$\Phi$ defined by~$\phi$ is nonconstant on~$\mathcal{C}$. That is, if there exists at least one $\sigma\in\mathcal{C}$ such that $\Phi(\sigma^\star)$ is strictly less than $\Phi(\sigma)$. This occurs only when~$\sigma^\star$ lies on the boundary of~$\mathcal{C}$, which is proved in the following corollary.  

Here, we mean that an element~$\sigma\in\mathcal{C}$ is in the interior of~$\mathcal{C}$ if for all $\sigma'\in\mathcal{C}$ there exists a $t<0$ with $\abs{t}$ small enough such that $\sigma+t(\sigma'-\sigma)$ is in~$\mathcal{C}$, and~$\sigma$ is on the boundary of~$\mathcal{C}$ otherwise. If the convex subset $\mathcal{C}$ is full-dimensional in $\H_n$, then these notions coincide with the standard definitions of the interior and boundary of $\mathcal{C}$. Otherwise, this coincides with the notion of the \emph{relative interior} and \emph{relative boundary} (see~\cite[p. 66]{Borwein2010}).

\begin{corollary}
\label{cor:boundary}
Let $g$, $\rho$, $f_\rho$ and $\mathcal{C}$ be defined as above  such that $\sigma^\star\in\mathcal{C}$ optimizes $f_\rho$ over~$\mathcal{C}$.  If the the linear functional $\Phi:\mathcal{C}\rightarrow\mathbb{R}$ defined by $\Phi(\sigma)=\Tr[\phi\sigma]$
is nonconstant on~$\mathcal{C}$, where $\phi=D_{g,\sigma^\star}(\rho)$, then $\sigma^\star$ is a boundary point of $\mathcal{C}$.
\end{corollary}
\begin{proof}
Suppose~$\sigma^\star$ is in the interior of~$\mathcal{C}$.  Since~$\sigma^\star$ is optimal, Theorem~\ref{thm:maintheorem_secondary} implies that $\Tr[\phi(\sigma-\sigma^\star)]\leq 0$ for all~$\sigma\in\mathcal{C}$. However, for each $\sigma$ there exists a $t<0$ with~$\abs{t}$ small enough such that $\sigma^\star+t(\sigma-\sigma^\star)$ is also in~$\mathcal{C}$. So we must also have \mbox{$\Tr[\phi(\sigma-\sigma^\star)]\geq 0$} for all $\sigma\in\mathcal{C}$. Hence $\Tr[\phi(\sigma-\sigma^\star)]=0$ for all $\sigma\in\mathcal{C}$, which proves the corollary.
\end{proof}

If~$\sigma^\star$ is on the boundary of~$\mathcal{C}$, then characterizing all of the supporting functionals of~$\mathcal{C}$ that are maximized by~$\sigma^\star$ allows us to find all matrices~$\rho$ for which~$\sigma^\star$ minimizes the function $f_\rho f(\sigma)=-\Tr[\rho g(\sigma)]$ over~$\mathcal{C}$.  
\begin{corollary}
Let $g$, $\rho$, $f_\rho$ and $\mathcal{C}$ be defined as above and $\sigma^\star$ be a boundary point of $\mathcal{C}$. Assume furthermore that $g'(\lambda)\neq 0$ for all eigenvalues $\lambda\in(a,b)$ of~$\sigma^\star$.  Then $f_\rho$ achieves a minimum at~$\sigma^\star$ if and only of~$\rho$ is of the form
\begin{equation}
 \rho=D_{g,\sigma^\star}^{\ddagger}(\phi)
\end{equation}
as long as $0\leq D_{g,\sigma^\star}^{\ddagger}(\phi)$, where $\phi\in \H_n$ is zero outside of the support of $D_{g,\sigma^\star}$ and defines a supporting functional of~$\mathcal{C}$ that is maximized by~$\sigma^\star$, i.e
\begin{equation}
\label{eq:SepHyp}
 \Tr[\phi\sigma]\leq \Tr[\phi\sigma^\star]
\end{equation}
for all $\sigma\in\mathcal{C}$.
\end{corollary}
The requirement that $\phi$ be zero outside of the support of $D_{g,\sigma^\star}$ means that $\bra{\psi}\phi\ket{\psi}=0$ for every eigenvector~$\ket{\psi}$ of $\sigma^\star$ with corresponding eigenvalue $\lambda\not\in(a,b)$. 
\begin{proof}
Suppose~$\sigma^\star$ minimizes $f_\rho$ over~$\mathcal{C}$. Then the matrix~$\phi=D_{g,\sigma^\star}(\rho)$ is zero outside of the support of $D_{g,\sigma^\star}$ by definition of $D_{g,\sigma^\star}$, and $\phi$ defines a supporting functional of~$\mathcal{C}$ by Theorem \ref{thm:maintheorem_secondary}. Since $f_\rho(\sigma^\star)$ is finite, $\rho$ must be zero outside of the support of~$D_{g,\sigma^\star}$, so we have that $D_{g,\sigma^\star}^\ddagger( D_{g,\sigma^\star}(\rho))=\rho$. Thus $\rho=D_{g,\sigma^\star}^{\ddagger}(\phi)$ as desired. 

Now suppose that $\phi\in \H_n$ is a matrix that defines a supporting functional of~$\mathcal{C}$ satisfying eq. \eqref{eq:SepHyp} for all $\sigma\in\mathcal{C}$ and is zero outside of the support of~$D_{g,\sigma^\star}$. If \mbox{$\rho=D_{g,\sigma^\star}^{\ddagger}(\phi)\geq 0$}, then $f_\rho$ is convex and the matrix \mbox{$\phi=D_{g,\sigma^\star}(\rho)$} is of the desired form. Hence, by Theorem~\ref{thm:maintheorem_secondary}, we have that~$\sigma^\star$ minimizes $f_\rho$ as desired.
\end{proof}

\section{Relative entropy of entanglement}
\label{sec:REE}
For a quantum state~$\rho$, the relative entropy of entanglement is a quantity that may be defined in terms of a convex optimization problem.  The function that is to be optimized is the relative entropy $S(\rho\| \sigma)$, defined as
\begin{equation}
 S(\rho\| \sigma)=-S(\rho)-\Tr[\rho\log\sigma],
\end{equation}
and $S(\rho)=-\Tr[\rho\log\rho]$ is the von Neuman entropy. For $\rho\in \H_{n,+,1}$ and $\sigma\in \H_{n,+}$, the relative entropy has the important properties that $S(\rho\| \sigma)\geq0$ and $S(\rho\| \sigma)=0$ if and only if $\rho=\sigma$.  We can extend the range of $S(\rho\| \sigma)$ to $[0,+\infty]$ such that $S(\rho\| \sigma)=+\infty$ if the matrix~$\rho$ is nonzero outside the support of~$\sigma$~\cite{Ohya1993}. That is, if $\bra{\psi}\rho\ket{\psi}>0$ for some $\ket{\psi}$ such that $\bra{\psi}\sigma^\star\ket{\psi}=0$. 

The \emph{relative entropy of entanglement} (REE) of a state $\rho\in \H_{n,+,1}$ was originally defined as
\begin{equation}
 E_{\mathcal{D}}(\rho)=\min_{\sigma\in \mathcal{D}}S(\rho\| \sigma),
\end{equation}
where $\mathcal{D}\subset \H_{n,+,1}$ is the convex subset of separable states~\cite{Vedral1997,Vedral1998}. 
The REE has not only been shown to be a useful measure of entanglement, but its value comprises a computable lower bound to another important measure of entanglement, the distillable entanglement~\cite{Bennett1996,Rains1999a,Rains1999}, whose optimization over purification protocols is much more difficult than the convex optimization required to calculate the REE.  Additionally, the regularized version of the REE, defined as
\begin{equation*}
 E_{\mathcal{D}}^\infty(\rho)=\lim_{k\rightarrow\infty}\frac{1}{k}E_{\mathcal{D}}(\rho^{\otimes k} ),
\end{equation*}
plays a role analogous to entropy in thermodynamics~\cite{Brandao2010} and gives an improved upper bound to the distillable entanglement~\cite{Rains1999a}.
Recall that $ E_{\mathcal{D}}^\infty(\rho)\le  E_{\mathcal{D}}(\rho)$.
Unfortunately, the computational complexity of computing the REE is high, since it is difficult to characterize when a state is separable~\cite{Gurvits2003}.

 Alternatively, the relative entropy of entanglement can be defined as
\begin{equation}
 E_{\mathcal{P}}(\rho)=\min_{\sigma\in \mathcal{P}}S(\rho\| \sigma),
\end{equation}
where the optimization is instead taken over the convex set of states that are positive under partial transposition (PPT),  denoted by $\mathcal{P}$~\cite{Rains1999a,Rains1999}. This definition of the REE, as well as its regularized version $E_\mathcal{P}^\infty$, are also upper bounds to the distillable entanglement.  Since $\mathcal{P}$ includes the separable states of quantum systems of any dimension~\cite{Peres1996}, the quantities $E_{\mathcal{P}}$ and $E_\mathcal{P}^\infty$ are smaller than their $\mathcal{D}$-based counterparts, so they offer improved bounds to the distillable entanglement. Because of this fact, along with the fact that the set of PPT states is much easier to characterize than the set of separable ones, 
i.e. polynomially in the dimension of the state,
we will primarily take $E_\mathcal{P}$ to be the REE for the remainder of this paper.

Although we limit ourselves here to consideration of the sets of separable and PPT states, it is also possible to define a relative entropy with respect to \emph{any} convex set of positive operators. That is, given any convex subset $\mathcal{C}\subset\H_{n,+}$, we can define the quantity
\[
 E_{\mathcal{C}}(\rho)=\min_{\sigma\in\mathcal{C}} S(\rho\| \sigma).
\]
Such quantities are useful in generalized resource theories in quantum information~\cite{Gour2007,Brand2010,Narasimhachar2013}. In such cases, the set of ``free'' states that may be used in a resource theory comprises a convex set and the resourcefulness of a given state may be measured by its relative entropy to the set of free states. Many of the results derived here for $E_\mathcal{P}$ also hold true for relative entropies with respect to arbitrary convex sets $\mathcal{C}$.

There is no systematic method for calculating the REE with respect to either $\mathcal{D}$ or $\mathcal{P}$, even for pure states, so it is  worthwhile to seek cases for which an explicit expression for the REE may be obtained.  From a given state~$\sigma^\star$, a method to determine all entangled states $\rho\not\in\mathcal{P}$ whose REE may be given by $E_\mathcal{P}(\rho)=S(\rho\| \sigma^\star)$ has been previously developed~\cite{Ishizaka2003,Friedland2011}. In the following, using our notation from section \ref{sec:maintheorem}, we restate the results from Friedland and Gour~\cite{Friedland2011} and omit the proofs.

\subsection{Derivative of \texorpdfstring{$\Tr[\rho\log(A+tB)]$}{Tr[plog(A+tB)}}
For a fixed quantum state $\rho\in \H_{n,+,1}$, the relative entropy $S(\rho\| \sigma)$ is minimized over $\mathcal{P}$ whenever the function $f_\rho(\sigma)=-\Tr[\rho\log(\sigma)]$ is minimized. Thus, given a state~$\sigma^\star$ on the boundary of $\mathcal{P}$ and taking $g(x)=\log(x)$ (which is operator concave~\cite{Carlen2010}), we may use the analysis from section \ref{sec:maintheorem} in order to find states~$\rho$ such that the relative entropy $S(\rho\| \sigma)=f_\rho(\sigma)$ is minimized by~$\sigma^\star$.  

Note that $g(x)=\log x$ is analytic on $(0,+\infty)$. Thus, for any matrix $A\in\H_{n,+}$ that is zero outside the support of $\rho$, the directional derivative can be given as \mbox{$f_\rho'(A;B)=D_{\log,A}(B)$}, and this derivative exists for all \mbox{$B\in\H_n$}. For simplicity, we define $L_A=D_{\log,A}$ as well as $S_A=S_{\log,A}$.
 Noting that $\frac{d}{dx}\log(x)=\frac{1}{x}$, the corresponding matrix $S_A$ is given by
\[
[S_A]_{ij}=\left\{\begin{array}{ll}
          \frac{a_i-a_j}{\log(a_i)-\log(a_j)} & a_i\neq a_j,\,a_i,a_j>0\\
          a_i& a_i=a_j>0\\
          0 & a_i=0 \textrm{ or }a_j=0.
         \end{array}
\right. 
\]
Finally, we note that $L_A(A)=P_A$ and also $L_A^\ddagger(\mathds{1})=L_A^\ddagger(P_A)=A$ for all $A\in \H_{n,+}$. This is a special property of the choice of function $g(x)=\log x$ that makes the relative entropy easy to study compared to arbitrary concave functions~$g$.

\subsection{A criterion for closest \texorpdfstring{$\mathcal{P}$}{P}-states}
\label{sec:CSScriterion}
A state $\sigma^\star\in\mathcal{P}$ that minimizes the relative entropy with~$\rho$, i.e. $E_\mathcal{P}(\rho)=S(\rho\| \sigma^\star)$, is said to be a \emph{closest \mbox{$\mathcal{P}$-state}} (C$\mathcal{P}$S) to~$\rho$. Since $S(\rho\| \sigma)\geq 0$ for all matrices $\rho,\sigma\in \H_{n,+}$, and $S(\rho\| \sigma)=0$ if and only if $\rho=\sigma$, the REE of any state $\rho\in\mathcal{P}$ vanishes.  Thus, the more interesting cases to analyze occur when~$\rho\not\in\mathcal{P}$. 

Before finding all states $\rho\not\in\mathcal{P}$ that have the given state $\sigma^\star\in\mathcal{P}$ as a C$\mathcal{P}$S, we present some useful facts. Note that the function $f_\rho$ is \emph{proper} on $\mathcal{P}$ for all states~$\rho$. That is, there always exists a state $\sigma\in\mathcal{P}$ such that $f_\rho(\sigma)$ is finite. Indeed, the maximally mixed state $\tfrac{1}{n}\mathds{1}$ is always in $\mathcal{P}$ and $\Tr[\rho\log(\tfrac{1}{n}\mathds{1})]$ is finite for all $\rho\in\H_{n,+}$.  If a state $\sigma^\star\in\mathcal{P}$ is a C$\mathcal{P}$S to a state $\rho\not\in\mathcal{P}$, and $\sigma^\star$ is not full-rank, then $\rho$ must be zero outside of the support of~$\sigma^\star$. Otherwise, $f_\rho(\sigma^\star)$ is not finite, and thus $\sigma^\star$ is not optimal. Furthermore, since $\mathcal{P}$ is compact and $f_\rho$ is bounded below, $f_\rho$ attains its minimum value over $\mathcal{P}$ for all states $\rho$. 

We also note that a C$\mathcal{P}$S of any non-PPT state $\rho\in\mathcal{P}$ must be on the boundary of $\mathcal{P}$ (see~\cite[Cor. 2]{Friedland2011}). This agrees with the notion of the relative entropy as being a distance-like measure on the space of states. Note this fact is unique to the choice $g(x)=\log x$ in the minimization target \mbox{$f_\rho(\sigma)=-\Tr[\rho\, g(\sigma)]$}, since in this case we have $\frac{d}{dx}\log x=\frac{1}{x}$ and thus $L_{A}(A)=P_A$ for all matrices $A\geq 0$. 

We are now ready to state the criterion for a state \mbox{$\rho\not\in\mathcal{P}$} to have~$\sigma^\star$ as a C$\mathcal{P}$S. Since~$\sigma^\star$ is on the boundary of $\mathcal{P}$, there is a proper supporting functional at~$\sigma^\star$ defined by a matrix $\phi\in \H_n$ of the form $ \Tr[\phi\sigma] \leq \Tr[\phi\sigma^\star]=1$. Here the condition that $\phi$ be \emph{proper} means that there exists at least one $\sigma\in\mathcal{P}$ that is zero outside the support of~$\sigma^\star$ with $\Tr[\phi\sigma]<1$. Hence $\phi\neq P_{\sigma^\star}$, so we can restrict to supporting hyperplanes such that $\Tr[(P_{\sigma^\star}-\phi)^2]\neq 0$. For each such $\phi$, we construct the family of states
 \begin{equation}
 \label{eq:rhoform}
  \rho(\sigma^\star,\phi,x)=(1-x)\sigma^\star+xL_{\sigma^\star}^{\ddagger}(\phi), \, \,\,
  x\in(0,x_{\max}],
 \end{equation}
where, for a given~$\sigma^\star$ and supporting functional defined by $\phi$, the value $x_{\max}$ is the largest such that $\rho(\sigma^\star,\phi,x_{\max})$ has no negative eigenvalues and is thus a valid quantum state.  We consider case when $\sigma^\star$ is singular separately from when $\sigma^\star$ is full-rank. 

\begin{theorem}[Full-rank C$\mathcal{P}$S]
 Let~$\sigma^\star$ be a full-rank state be on the boundary of $\in \mathcal{P}$. Then~$\sigma^\star$ is a C$\mathcal{P}$S to a state $\rho\not\in\mathcal{P}$ if and only if~$\rho$ is of the form $\rho=\rho(\sigma^\star,\phi,x)$ in~\eqref{eq:rhoform}, where $\phi$ defines a supporting functional of $\mathcal{P}$ at~$\sigma^\star$ such that
\begin{equation}\label{eq:thmCPS_supportingfunctional}
 \Tr[\phi\sigma]\leq\Tr[\phi\sigma^\star]=1 \hspace{2mm}\text{for all } \sigma\in\mathcal{P}
\end{equation}
and $\Tr[(\mathds{1}-\phi)^2]=1$, and $x\in(0,x_{\max}]$ where $x_{\max}$ is the largest value of $x$ such that $\rho(\sigma^\star,\phi,x)$ is a positive semi-definite matrix. 
\end{theorem}

The proof of this (see~\cite[Thm. 3]{Friedland2011}) relies on the fact that $P_{\sigma^\star}=\mathds{1}$ if $\sigma^\star$ is full-rank. Namely, if a matrix $\phi\in H_n$ defines a proper supporting functional of $\mathcal{P}$ at $\sigma^\star>0$, then so does $\phi(t)=t\phi+(1-t)P_{\sigma^\star}$ for all $t>0$. In the case of singular $\sigma^\star$, this is only true in general for $t\in (0,1]$. Thus, only a sufficient condition regarding when a state $\rho$ has~$\sigma^\star$ as a closest $\mathcal{P}$-state may be given (see~\cite[Thm. 7]{Friedland2011}).

\begin{theorem}[Singular C$\mathcal{P}$S]
 Let~$\sigma^\star$ be a singular state be on the boundary of $\in \mathcal{P}$.  Then~$\sigma^\star$ is a C$\mathcal{P}$S to all states of the form $\rho=\rho(\sigma^\star,\phi,x)$ in \eqref{eq:rhoform}, where $\phi$ defines a supporting functional of $\mathcal{P}$ at~$\sigma^\star$ of the form in \eqref{eq:thmCPS_supportingfunctional} such that $\Tr[(P_{\sigma^\star}-\phi)^2]=1$, and $x\in(0,\tilde x_{\max}]$ where $  \tilde x_{\max} = \min\{1,x_{\max}\}$ and $x_{\max}$ is the largest value of $x$ such that $\rho(\sigma^\star,\phi,x)$ is a positive semi-definite matrix. 
\end{theorem}

For a given~$\sigma^\star$ on the boundary of $\mathcal{P}$, these conditions allow us to construct families of non-PPT states that have~$\sigma^\star$ as a C$\mathcal{P}$S. Indeed, each matrix~$\phi$ defining a supporting functional of the form in eq.~\eqref{eq:thmCPS_supportingfunctional} with $\Tr[(P_{\sigma^\star}-\phi)^2]=1$  determines a different family.  Fortunately, the boundary and the supporting functionals of the set of PPT matrices are easy to characterize. All matrices $\phi\in\H_n$ that define supporting functionals of $\mathcal{P}$ are of the form
\begin{equation}
\label{eq:PPThyperplane}
 \phi=\mathds{1}-\sum_{i\,:\,\mu_i=0}a_i\ket{\varphi_i}\bra{\varphi_i}^\Gamma
\end{equation}
with each $a_i\geq0$, where $\sigma^{\star\Gamma}=\sum\mu_i\ket{\varphi_i}\bra{\varphi_i}$ is the spectral decomposition of the partial transpose of~$\sigma^\star$. The partial transpose operation is a self-adjoint linear operator on matrices with respect to the Hilbert-Schmidt inner product, so for each $\sigma\in\mathcal{P}$
\begin{align*}
 \Tr[\ket{\varphi_i}\bra{\varphi_i}^\Gamma\sigma] &= \bra{\varphi_i}\sigma^\Gamma\ket{\varphi_i}\geq 0,
\end{align*}
since $\sigma^\Gamma\geq0$, and so $\Tr[\phi\sigma]\leq 1$ for all $\sigma\in\mathcal{P}$, and $\Tr[\phi\sigma^\star]=1$. 

If~$\sigma^\star$ is full-rank, then~$\sigma^\star$ is on the boundary of $\mathcal{P}$ if its partial transpose $\sigma^{\star\Gamma}$ has at least one zero eigenvector. The supporting functional of  $\mathcal{P}$ at~$\sigma^\star$ of this form is unique if $\sigma^{\star\Gamma}$ has exactly one zero eigenvector. Otherwise there exists a range of supporting functionals of $\mathcal{P}$ at~$\sigma^\star$ and thus there is a range of families of entangled states $\rho(\sigma^\star,\phi,x)$ that have~$\sigma^\star$ as a C$\mathcal{P}$S. 

Furthermore, if $\rho\not\in\mathcal{P}$ is full-rank, then its C$\mathcal{P}$S is unique. This fact is due to the strong concavity of the matrix function $\log$ and is proven in~\cite{Friedland2011}. 

For a state $\rho\not\in\mathcal{P}$ that has $\sigma^\star\in\mathcal{P}$ as a closest \mbox{$\mathcal{P}$-state}, a closed form expression for the relative entropy of entanglement can be given by
\begin{equation}
 E_\mathcal{P}(\rho)
 =-S(\rho)-\Tr[\phi(x) \sigma^\star \log(\sigma^\star)]
\end{equation}
where $\phi(x)=(1-x)\mathds{1}+x\phi$ and $\rho=\rho(\sigma^\star,\phi,x)$ is of the form in~\eqref{eq:rhoform}.


\subsection{Additivity of the relative entropy of entanglement}
It is of great importance to determine conditions for when the relative entropy of entanglement is weakly additive. These are states $\rho$ for which the relative entropy of entanglement is equal to its regularized version,
\[
 E_\mathcal{P}(\rho) = E_\mathcal{P}^\infty(\rho).
\]
Indeed, the regularized version has been shown to be the unique measure of entanglement in a reversible theory of entanglement~\cite{Brand2010}, although it is much more difficult to compute. Given a state $\rho$ and a PPT state $\sigma^\star$ such that $\sigma^\star$ is a C$\mathcal{P}$S to~$\rho$, it has been shown~\cite{Rains1999,Rains2000} that $E_\mathcal{P}(\rho)$ is weakly additive if 
\begin{equation}\label{eq:commutativity}[\rho,\sigma^\star]=0
\hspace{5mm}\text{and}\hspace{5mm}
\left(\rho{\sigma^{\star}}^{-1}\right)^\Gamma\geq \mathds{1}.
\end{equation} 
However, the state $\rho$ may be given by $\rho=(1-x)\sigma^\star+xL_{\sigma^\star}^{\ddagger}(\phi)$ for some matrix $\phi$ that defines a supporting functional. So the condition in \eqref{eq:commutativity} may be reduced to 
\begin{equation*}[L_{\sigma^\star}^{\ddagger}(\phi),\sigma^\star]=0
\hspace{5mm}\text{and}\hspace{5mm}
\left(L_{\sigma^\star}^{\ddagger}(\phi){\sigma^{\star}}^{-1}\right)^\Gamma\geq \mathds{1}.
\end{equation*}
However, we have $[L_{\sigma^\star}^{\ddagger}(\phi),\sigma^\star]=0$ if and only if $[\phi,\sigma^\star]=0$, so the conditions for weak additivity can be stated as
\begin{equation}\label{eq:commutativity2}[\phi,\sigma^\star]=0
\hspace{5mm}\text{and}\hspace{5mm}
\left(L_{\sigma^\star}^{\ddagger}(\phi){\sigma^{\star}}^{-1}\right)^\Gamma\geq \mathds{1}.
\end{equation} 
Hence, we have reduced the task of finding states $\rho$ for which  $E_\mathcal{P}(\rho)$ is weakly additive to finding states $\sigma^\star$ on the boundary of $\mathcal{P}$ and a corresponding supporting hyperplane of $\mathcal{P}$ at $\sigma^\star$ defined by $\phi$ that satisfy the conditions in \eqref{eq:commutativity2}. In particular, in the two-qubit case, it has been shown~\cite{Miranowicz2008} that  $E_\mathcal{P}(\rho)$ is weakly additive for all states $\rho$ that commute with a C$\mathcal{P}$S state $\sigma^\star$.

\section{Rains bound}
\label{sec:Rains}
\label{sec:rains}
Although the results in the previous section regarding the relative entropy of entanglement have been shown previously~\cite{Friedland2011}, we present here new results using these methods to investigate a similar quantity, the so-called \emph{Rains bound}~\cite{Rains1999,Rains2001}. This quantity was originally defined as
\begin{equation}
\label{eq:RainsOrigin}
 R(\rho)=\min_{\sigma}S(\rho\| \sigma)+\log\Tr\abs{\sigma^\Gamma},
\end{equation}
where the minimization is taken over \emph{all} normalized states~$\sigma\in \H_{n,+,1}$ rather than just the PPT states. While still an upper bound to the distillible entanglement, the Rains bound is also bounded above by $E_\mathcal{D}$, so it is an improved bound over the relative entropy of entanglement.

The function that is optimized in~\eqref{eq:RainsOrigin} is the sum of the relative entropy of~$\rho$ and~$\sigma$ and the \emph{logarithmic negativity}~\cite{Plenio2005}, defined as 
\[
 LN(\sigma)=\log\Tr\abs{\sigma^\Gamma}.
\]
 For a given matrix $A\in M_n$, we have~$\abs{A}=\sqrt{A^\dagger A}$ and the trace of~$\abs{A}$ is equal to the sum of the singular values of~$A$~\cite{Bhatia1997,Horn1991}. If $A$ is hermitian, than the singular values of $A$ are the absolute values of the eigenvalues of $A$. If a state~$\rho$ is positive under partial transposition then its logarithmic negativity vanishes, since $\Tr\abs{\rho^\Gamma}=\Tr[\rho]=1$ if $\rho^\Gamma=\abs{\rho^\Gamma}$. Thus the Rains bound itself vanishes on the PPT states. As a function of~$\rho$, the logarithmic negativity is an entanglement measure with the peculiar property that it is not convex~\cite{Plenio2005}. Convexity is needed, however, to be able to perform convex optimization analysis.

It was subsequently pointed out~\cite{Audenaert2002} that the evaluation of the Rains bound can be recast in terms of a convex optimization problem. We define the convex set~$\mathcal{T}\subset \H_{n,+}$ as
\begin{equation}
 \mathcal{T}=\left\{\tau\in \H_{n,+}\,\big|\,\Tr\abs{\tau^\Gamma}\leq 1\right\}.
\end{equation}
 The matrices in $\mathcal{T}$ represent \emph{subnormalized} states, since $\Tr[\tau]\leq\Tr\abs{\tau^\Gamma}\leq 1$ for each $\tau\in\mathcal{T}$, and a matrix $\tau\in\mathcal{T}$ has $\Tr[\tau]=1$ if and only if $\tau^\Gamma\geq0$, i.e. $\tau$ is PPT.  Furthermore, the set $\mathcal{T}$ is indeed convex since, for the convex combinations of two matrices $\tau,\tau'\in\mathcal{T}$, we have that $\Tr\abs{(1-t)\tau^\Gamma+t\tau'^\Gamma}\leq (1-t)\Tr\abs{\tau^\Gamma}+t\Tr\abs{\tau'^\Gamma}\leq 1$ for all $t\in[0,1]$, and thus $(1-t)\tau+t\tau'\in\mathcal{T}$.

With this, the Rains bound can then be restated in a manner analogous to $E_\mathcal{D}$ and $E_\mathcal{P}$ as
\begin{equation}
 R(\rho)=E_\mathcal{T}=\min_{\tau\in\mathcal{T}}S(\rho\| \tau).
\end{equation}
For a fixed~$\rho$, this is indeed a convex function of $\tau$ since~$S(\rho\| \tau)$ is jointly convex in its arguments regardless of the normalization of~$\rho$ and $\tau$~\cite{Ohya1993}.  We can now address the question as to when a matrix $\tau^\star\in\mathcal{T}$ minimizes the Rains bound for a given state~$\rho$.  That is,
\begin{equation*}
 S(\rho\| \tau^\star)\leq  S(\rho\| \tau) \text{ for all }\tau\in \mathcal{T},
 \label{eq:closestP2}
\end{equation*}
so that $\tau^\star$ satisfies $R(\rho)=S(\rho\| \tau^\star)$.   For a given~$\rho$, we investigate the necessary and sufficient conditions that must be satisfied for $\tau^\star$ to minimize the Rains bound for~$\rho$.

Using the same methods from section \ref{sec:REE} for determining the conditions for when a matrix $\sigma^\star \in\mathcal{P}$ optimizes the REE, namely by characterizing the supporting functionals of a convex set, similar conditions can be shown for the minimization of the Rains bound. The requirements for a matrix $\phi\in\H_{n}$ to define a supporting functional are outlined in section~\ref{sec:RainsSepHypPlane}, while the necessary and sufficient conditions for a matrix $\tau^\star\in\mathcal{T}$ to minimize $E_\mathcal{T}(\rho)$ for a state $\rho\in\H_{n,+,1}$ are given in section~\ref{sec:RainsConditions}.  We subsequently use these results to compare the REE and the Rains bound in section~\ref{sec:REEcompareRains}.

\subsection{Supporting functionals of \texorpdfstring{$\mathcal{T}$}{T}}
\label{sec:RainsSepHypPlane}

Recall that, for a matrix $\tau^\star\in\mathcal{T}$, finding all matrices~$\rho\in\H_{n,+}$ such that $\tau^\star$ minimizes $f_\rho(\tau)=-\Tr[\rho g(\tau)]$ over $\mathcal{T}$ is reduced to finding the supporting functionals of $\mathcal{T}$ at $\tau^\star$. These are given by a matrix  $\phi\in\H_{n}$ such that
\[
 \Tr[\phi\tau]\leq \Tr[\phi\tau^\star] \,\,\textrm{ for all }\tau\in\mathcal{T}.
\]
Note that for each $\tau\in\mathcal{T}$, its partial transpose $\tau^\Gamma$ is an element of the unit ball $B_n^1\subset\H_n$
\begin{equation*}
 B_n^1= \left\{\alpha\in\H_n\,\big|\,\norm{\alpha}_1\leq 1\right\},
\end{equation*} 
with respect to the Schatten-1 norm  $\norm{\alpha}_1=\Tr\abs{\alpha}$~\cite{Bhatia1997}. The unit ball is also a convex set, and characterizing all supporting functionals of $B_n^1$ at will assist us in finding the supporting functionals of $\mathcal{T}$. In particular, each matrix $\omega\in\H_n$ whose maximum singular value is equal to 1 defines a supporting functional of $B_n^1$, since
\[
 \Tr[\omega\alpha]\leq 1 \hspace{5mm}\text{for all }\alpha\in B_n^1,
\]
and $\Tr[\omega\alpha]$ achieves the maximum value of 1 for some $\alpha$ on the boundary of $B_n^1$~\cite{Horn1985}. 

Recall that a matrix $P\in\H_n$ is called an \emph{orthogonal projection} if $P^2=P$, and that two orthogonal projections $P_1$ and $P_2$ are said to be \emph{disjoint} if $P_1P_2=P_2P_1=0$, i.e. the intersections of the ranges of $P_1$ and $P_2$ is the zero vector.  Note that the sum of the ranks of two orthogonal projections $P_1$ and $P_2$ is at most $n$.
\begin{lemma}\label{lemma:projectors}
 Let $\alpha\in \H_{n}$ and let~$P_1$ and~$P_2$ be orthogonal projections with ranks $p_1$ and $p_2$ respectively.  Then, for $p=p_1+p_2$,
 \begin{equation}
   \Tr(P_1 \alpha )- \Tr(P_2 \alpha )\leq \sum_{i=1}^p s_i(\alpha),
  \label{eq:projectors}
 \end{equation}
 where $s_i(\alpha)$ are the singular values of $\alpha$ listed in decreasing order.   Equality in \eqref{eq:projectors} holds if and only if~$P_1$ and~$P_2$ are projections onto subspaces corresponding to the nonnegative and nonpositive eigenvalues of $\alpha$ respectively and the absolute values of the  corresponding eigenvalues are the largest $p$ singular values of $\alpha$.
\end{lemma}
\begin{proof}
 See, e.g., Thm. 3.4.1, p. 195, in~\cite{Horn1991}.
\end{proof}
Suppose that $\norm{\alpha}_1=1$, that $\alpha$ has at most $p$ nonzero singular values and that~$P_1$ and~$P_2$ are projections onto subspaces corresponding to the nonnegative and nonpositive eigenvalues of $\alpha$. Then the matrix $\phi=P_1-P_2$ defines a supporting functional of $B_n^1$ at $\alpha$. Indeed, we have $\Tr[\phi \alpha]=1$ and $\Tr[\phi \beta]\leq \norm{\beta}_1\leq 1$ for all other $\beta\in B_n^1$.

\begin{theorem}
 Let $\alpha\in\H_{n}$ be a matrix on the boundary of~$B_n^1$, namely $\norm{\alpha}_1=\sum^n_{i=1}s_i(\alpha)=1$.  A linear functional $\Phi:H_n\rightarrow\RR$ given by $\Phi(\beta)=\Tr[\phi\beta]$ for some matrix $\phi\in\H_n$ is a supporting functional of $B_n^1\subset\H_n$ at~$\alpha$  such that
 \begin{equation}
 \label{eq:rainsbound_Bsuppfunctional}
\Tr[\phi\beta]\leq\Tr[\phi\alpha]=1\,\,\,\textrm{ for all }\beta\in B_n^1,                                                                                                                                                                                                                                                                                                                                                                                                                 \end{equation}
if and only if $\phi$ is a convex combination of matrices of the form $P_1-P_2$ such that
 \[
  \Tr[(P_1-P_2)\beta]\leq \Tr[(P_1-P_2)\alpha]=1,
 \]
where~$P_1$ and~$P_2$ are disjoint orthogonal projection matrices as in Lemma \ref{lemma:projectors}. In particular, the supporting functional at $\alpha$ of the form in \eqref{eq:rainsbound_Bsuppfunctional} is unique if and only if $\alpha$ has no zero eigenvalues.
\end{theorem}
For a given $\alpha$, the ranks of~$P_1$ and~$P_2$ must have ranks at least as large as the number of positive and negative eigenvalues of $\alpha$ respectively.  
\begin{proof}
From Lemma \ref{lemma:projectors}, we see that for any commuting~$P_1$ and~$P_2$  of the form above, $P_1-P_2$ does indeed define a supporting functional at $\alpha$, since $\Tr[(P_1-P_2)\alpha]=\sum_i s_i(\alpha)=1$ and, from Lemma~\ref{lemma:projectors},
\[
 \Tr[\beta(P_1-P_2)]\leq \norm{\beta}_1 = 1 
\]
for each $\beta\in B_n^1$. For any convex combination of hyperplanes of this form, i.e. for pairs of commuting projection matrices $P_1^{(i)}$ and $P_2^{(i)}$  of the required form such that $\omega=\sum_i q_i(P_1^{(i)}-P_2^{(i)})$ for $0\leq q_i\leq 1$ and $\sum_iq_i=1$, we have
\begin{equation*}
 \Tr[ \omega\beta]=\sum_i q_i \Tr[\beta(P_1^{(i)}-P_2^{(i)})] \leq \sum_i q_i = 1 = \Tr[ \omega\alpha].
\end{equation*}

Now consider any matrix $\omega$ that defines a supporting functional of $B_n^1$ at $\alpha$, i.e. any $\omega\in\H_n$ such that $\Tr[\omega\beta]\leq  \Tr[\omega\alpha]=1$ for all $\beta\in B_n^1$.  In the eigenbasis of $\alpha$, the matrix $\alpha$ is the block diagonal matrix \mbox{$\alpha=\text{diag}(\alpha_1,\alpha_2,\alpha_3)$}, where $\alpha_1$ and $\alpha_2$ are diagonal matrices whose entries consist of the positive and negative eigenvalues of $\alpha$ respectively and $\alpha_3$ is the zero matrix with size equal to dimension of the nullspace of $\alpha$.  Then any supporting functional with $\Tr[\omega\alpha]=\Tr\abs{\alpha}=1$ must be a block diagonal matrix of the form $\omega=\text{diag}(\omega_1,\omega_2,\omega_3)$, where $\omega_1=\mathds{1}$, $\omega_2=-\mathds{1}$ and $\omega_3$ is any hermitian matrix with maximum absolute eigenvalue less that $1$.  It is straightforward to show that any such $\omega$ is a convex combination of matrices of the form $P_1-P_2$ as described above.

From this analysis, we see that $\omega$ is unique if and only if $\alpha$ has no zero eigenvalues.
\end{proof}

For a matrix $\tau\in\mathcal{T}$ with $\norm{\tau^\Gamma}_1=1$, we note that~$\omega^\Gamma$ defines a supporting functional of $\mathcal{T}$ at $\tau$ if $\omega\in\H_n$ is a supporting plane of $B_n^1$ at $\tau^\Gamma$. Indeed, since the partial transpose is self-adjoint with respect to the Hilbert-Schmidt inner product $\Tr[\omega\tau^\Gamma]=\Tr[\omega^\Gamma\tau]$, we have that $\Tr[\omega^\Gamma\tau]\leq\Tr[\omega^\Gamma\tau]$ for all $\tau\in\mathcal{T}$.

\subsection{A criterion for minimization of the Rains bound}
\label{sec:RainsConditions}

 Let $\rho\in \H_{n,+,1}$ and $\tau^\star\in\mathcal{T}$ such that $\tau^\star$ minimizes the Rains bound for~$\rho$.   If $\tau^\star$ is singular, then~$\rho$ is zero outside the support of $\tau^\star$.
This is straightforward and follows anagolously from a similar statement for the relative entropy of entanglement (see~\cite[Cor. 2]{Friedland2011}).  Consequently, if~$\rho$ is full-rank then~$\tau^\star$ must also be full-rank.

Since $R(\rho)$ vanishes for states $\rho\in\mathcal{P}$, the only interesting cases occur when~$\rho$ is not PPT. 
\begin{theorem}
 Let $\rho\in \H_{n,+,1}$ be a state that is not in~$\mathcal{P}$.  Then a matrix $\tau^\star\in\mathcal{T}$ satisfies $R(\rho)=S(\rho\| \tau^\star)$ if and only if $\norm{\tau^{\star\Gamma}}_1=1$ and, for all $\tau\in\mathcal{T}$,
 \begin{equation}
 \label{eq:RainsCondition}
  \Tr[L_{\tau^\star}(\rho)\tau]\leq\Tr[ L_{\tau^\star}(\rho)\tau^\star]=1 .
 \end{equation}
 Since $\norm{\tau^{\star\Gamma}}_1=1$ implies that $\tau^\star$ is on the boundary of $\mathcal{T}$, this means that $L_{\tau^\star}(\rho)$ defines a supporting functional of~$\mathcal{T}$ at $\tau^\star$.
\end{theorem}
\begin{proof}

Assume that $\tau^\star$ minimizes the Rains bound for~$\rho$.  We now show that $\norm{\tau^{\star\Gamma}}_1=1$, i.e. $\tau^\star$ is on the boundary of $\mathcal{T}$.  Indeed, if $\norm{\tau^{\star\Gamma}}_1<1$ then set $c=\tfrac{1}{\norm{\tau^{\star\Gamma}}_1}$.  Clearly~$c\tau^\star$ is in $\mathcal{T}$ and
\[
 \norm{(c\tau^\star)^\Gamma}_1= c\norm{\tau^{\star\Gamma}}_1=1.
\]
Since $c>1$, we have that $\Tr\left(\rho\log(c\tau^\star)\right)<\Tr\left(\rho\log\tau^\star\right)$ and so $S(\rho\| \tau^\star)<S(\rho\| c\tau^\star)$, a contradiction to $\tau^\star$ being optimal.  

Analogous to the case with the REE, the condition in~\eqref{eq:RainsCondition} is due to Theorem~\ref{thm:maintheorem_secondary}. Furthermore, since $L_{\tau^\star}(\tau^\star)=P_\tau^\star$ and the linear operator $L_{\tau^\star}$ is self-adjoint with respect to the trace, we have that $\Tr[L_{\tau^\star}(\rho)\tau^\star ]=\Tr[P_{\tau^\star}\rho]=1$, where~$\rho$ has unit trace and is zero outside of the support of~$\tau^\star$.  
\end{proof}

Given a matrix $\tau^\star\in\mathcal{T}$ with $\norm{\tau^{\star\Gamma}}_1=1$, it is now possible to find all states~$\rho$ for which $\tau^\star$ minimizes the Rains bound by finding all supporting functionals of $\mathcal{T}$ at $\tau^\star$. All of the supporting functionals considered here will have the form
\begin{equation}
\label{eq:standardSuppHyperplaneRains}
 \Tr[\phi\tau]\leq\Tr[\phi\tau^\star]=1
\end{equation}
for all $\tau\in\mathcal{T}$, such that $\phi$ defines a supporting functional of $\mathcal{T}$ at $\tau^\star\in\partial\mathcal{T}$.

\begin{theorem}
\label{thm:RainsBound}
 Let $\tau^\star\in\mathcal{T}$ with $\norm{\tau^{\star\Gamma}}_1=1$. Then there exists a state~$\rho$ such that $\tau^\star$ minimizes the Rains bound for~$\rho$ if and only if there exists a supporting functional of $\mathcal{T}$ at $\tau^\star$ defined by $\phi\in \H_{n}$ of the form in eq. \eqref{eq:standardSuppHyperplaneRains} such that
 \begin{equation}
 L^{\ddagger}_{\tau^\star}(\phi)\geq 0
 \end{equation}
and is zero outside of the support of $\tau^\star$ if $\tau^\star$ is singular, that is $\phi=P_{\tau^\star} \phi P_\tau^\star$, and~$\rho$ must be of the form \mbox{$\rho=L_{\tau^\star}^{\ddagger}(\phi)$}.  Furthermore, $\phi\geq0$ and, in particular, if $\rho>0$, then $\tau^\star$ is also full-rank and $\phi=L_{\tau^\star}(\rho)>0$.  
\end{theorem}
Recall that  $L_{\tau^\star}^\ddagger=L_{\tau^\star}^{-1}$  if $\tau^\star$ is full-rank.
\begin{proof}
First assume that there exists a~$\rho$ such that $\tau^\star$ minimizes the Rains bound for~$\rho$.  Then set~$\phi=L_{\tau^\star}(\rho)$.  We see that $\phi$ defines  a supporting functional of the form in eq. \eqref{eq:standardSuppHyperplaneRains} and satifies the required conditions.  Indeed, if~$\tau^\star$ is singular then $L_{\tau^\star}(\rho)$ is zero outside of the support of $\tau^\star$ by definition.  Moreover, $L_{\tau^\star}^\ddagger(\phi)=L_{\tau^\star}^\ddagger(L_{\tau^\star}(\rho))=\rho$ since~$\rho$ is zero outside of the support of $\tau^\star$. 

Now assume that there exists a hyperplane $\phi$ that satisfies the conditions and set $\rho\equiv L^{\ddagger}_{\tau^\star}(\phi)$.  Then~$\rho$ is indeed a valid state, since $\rho=L^{\ddagger}_{\tau^\star}(\phi)\geq 0$ and
\[
 \Tr\rho=\Tr[P_{\tau^\star}\rho]=\Tr[\tau^\star L_{\tau^\star}(\rho)]=\Tr[\tau^\star\phi]  = 1.
\]
Finally, $\tau^\star$ indeed minimizes the Rains bound for~$\rho$, since $\phi=L_{\tau^\star}(\rho)$ defines a supporting functional of the desired form.

We now check the positivity of $\phi=L_{ \tau^\star}(\rho)$. If a function $g:(a,b)\longrightarrow\RR$ is a matrix monotone, that is $A\leq B$ implies $g(A)\leq g(B)$ for all $A,B\in\H_n(a,b)$, it can be shown (see~\cite[Thm. 6.6.36]{Horn1991}) that~$0\leq T_{g,A}$ for each~\mbox{$A\in\H_n(a,b)$}.  

Since $g(x)=\log(x)$ is matrix monotone on $(0,\infty)$, the matrix $T_A$ is positive semidefinite for all \mbox{$A\geq0$}. Due to the Schur product theorem (see for example~\cite[Thm. 5.2.1]{Horn1991}), if $T_A\geq0$ then $T_A\circ B\geq0$ for any $B\geq0$. Thus $L_{\tau^\star}(\rho)\geq0$ for any~$\rho$ such that $\tau^\star$ minimizes the Rains bound. Furthermore, $T_A\circ B>0$ for any $B>0$ as long as all diagonal elements of $T_A$ are nonzero, and the diagonal elements of $T_{A}$ are $\tfrac{1}{a_i}\neq 0$ for $A>0$. Thus we have $L_{\tau^\star}(\rho)>0$ if~$\rho$ is full-rank since $\tau^\star$ must also be full-rank. 
\end{proof}

If $\tau^{\star\Gamma}$ is full-rank, then $\phi=(P_1-P_2)^\Gamma$ is the unique hyperplane of $\mathcal{T}$ at $\tau^\star$, where~$P_1$ and~$P_2$ are the projectors onto the positive and negative eigenspaces of $\tau^\star$ respectively.  
Thus, as long as the condition $L_{\tau^\star}^{-1}(\phi)\geq0$ is satisfied, there is a unique state $\rho=L_{\tau^\star}^{-1}(\phi)$ for which~$\tau^\star$ minimizes the Rains bound. For singular $\tau^\star$, there may not exist any matrices~$\phi$ that define supporting functionals of~$\mathcal{T}$ at~$\tau^\star$ that are zero outside of the support of~$\rho$. Thus, there are many matrices $\tau^\star\in\mathcal{T}$ with $\norm{\tau^{\star\Gamma}}_1=1$ that do not minimize the Rains bound for \emph{any} states~$\rho$.

Finally, given a state $\tau^\star\in\mathcal{T}$ with $\norm{\tau^{\star\Gamma}}_1=1$ and a matrix $\phi$ defining a supporting functional of the desired form, a closed formula for the Rains bound for the state $\rho=L_{\tau^\star}^\ddagger(\phi)$ can be given by
\begin{equation}
 R(\rho)=-S(\rho)-\Tr[\phi \tau^\star\log(\tau^\star)].
\end{equation}

\section{Comparing the REE and the Rains bound}
\label{sec:REEcompareRains}
\label{seccompare}

The closed-form expressions found in the previous sections reveal nontrivial information about the original quantities $E_\mathcal{P}(\rho)$ and $R(\rho)$ for arbitrary states~$\rho$.  It is of particular interest to find states~$\rho$ for which the Rains bound and the REE coincide such that conditions may be found to determine when $R(\rho)$ is indeed an improvement to the upper bound of the distillable entanglement given by $E_\mathcal{P}(\rho)$.  

\subsection{Bipartite systems where one subsystem is a qubit}

It has been shown in~\cite{Miranowicz2008} that, in the bipartite qubit case, the Rains bound is equivalent to the relative entropy of entanglement. We now generalize this result to the bipartite case where one of the subsystems is a qubit.

\begin{theorem}
 For states of bipartite systems where at least one of the subsystems is a qubit, the Rains bound and the relative entropy of entanglement are equivalent, i.e. $R(\rho)=E_\mathcal{P}(\rho)$  for all states~$\rho$.
\end{theorem}

\begin{proof}
 We may assume that $\rho>0$.  Indeed,  since $E_\mathcal{P}$ and $R$ are continuous~\cite{Rains2001}, by continuity this will be satisfied for all $\rho\geq 0$ as well.
 We will show that if $\tau^\star\in\mathcal{T}$  minimizes the Rains bound for~$\rho$, then~$\tau^{\star\Gamma}\geq 0$, that is $\tau^\star\in\mathcal{P}$ and thus $R(\rho)=E_\mathcal{P}(\rho)$.  Suppose that there exists a matrix $\tau^\star\in\mathcal{T}$ that $\tau^\star$ minimizes the Rains bound for~$\rho$ such that $\tau^{\star\Gamma}\not\geq 0$. Then we show a contradiction.
 
First consider the case where $\tau^{\star\Gamma}$ has only one negative eigenvalue and the remaining eigenvalues are positive.  If~$P_1$ and~$P_2$ are the projectors on the positive and negative eigenspaces of $\tau^{\star\Gamma}$, then $P_2$ is the rank-one projector on to the negative eigenvector of $\tau^{\star\Gamma}$.  Since $P_1+P_2=\mathds{1}$, there is a supporting functional defined by $\phi=(P_1-P_2)^\Gamma
=\mathds{1}-2P_2^\Gamma$.  We now show that the largest eigenvalue of~$P_2^\Gamma$ is greater than 1, and thus $\phi=\mathds{1}-2P_2^\Gamma$ is no longer positive. This is a contradiction to the previous theorem, namely that~$\phi>0$ if $\tau^\star$ minimizes the Rains bound for a state~$\rho>0$.

 The eigenvector corresponding to the negative eigenvalue of $\tau^{\star\Gamma}$ may be written as
 \[
  \ket{\psi_2}=\ket{0}\otimes\ket{u}+\ket{1}\otimes\ket{v}
 \]
where $\langle\psi_2|\psi_2\rangle=\langle u|u\rangle+\langle v|v\rangle=1$.  Thus, $P_2=\ket{\psi_2}\bra{\psi_2}$, 
and we may write $P_2$ and $P_2^\Gamma$ in matrix form as
\begin{equation*}
 P_2=\begin{pmatrix}
    \ket{u}\bra{u} & \ket{u}\bra{v}\\
    \ket{v}\bra{u} & \ket{v}\bra{v}
   \end{pmatrix}
   \hspace{2mm}\text{and}\hspace{2mm}
 P_2^\Gamma=\begin{pmatrix}
    \ket{u}\bra{u} & \ket{v}\bra{u}\\
    \ket{u}\bra{v} & \ket{v}\bra{v}
   \end{pmatrix}.
\end{equation*}
By the Cauchy interlacing theorem~\cite{Horn1985}, the maximal eigenvalue of $P_2^\Gamma$ is at least as great as the largest eigenvalue among the two submatrices $\ket{u}\bra{u}$ and $\ket{v}\bra{v}$.  These are both rank-one positive matrices, and the two eigenvalues of these submatrices must sum to unity. Thus, the largest eigenvalue of $P_2^\Gamma$ is at least as big as $\max(\langle u|u\rangle,\langle v|v\rangle)\geq\tfrac{1}{2}$, and so $\mathds{1}\not >2P_2^\Gamma$ as desired.

We now consider the general case when $\tau^{\star\Gamma}$  has at least one negative eigenvalue. Any supporting functional of~$\mathcal{T}$ at~$\tau^\star$ is defined by a matrix of the form \mbox{$\phi=(P_1-P_2+Q)^\Gamma$}. Here~$P_1$ and~$P_2$ are orthogonal projections onto the positive and negative eigenspaces of $\tau^{\star\Gamma}$ respectively, and $Q$ is a matrix such that $P_1Q=P_2Q=0$ and the largest absolute eigenvalue of $Q$ is at most 1.  Define $P_3$ as the projection onto the nullspace of $\tau^{\star\Gamma}$ such that  $P_3Q=QP_3=Q$.  Note that~$P_1,P_2$ and $P_3$ form a complete set of orthogonal projectors such that $P_1+P_2+P_3=\mathds{1}$.  Thus, we may write~$\phi^\Gamma$ as
\[
 \phi^\Gamma=P_1+P_2-Q = 
 \mathds{1}-F,
\]
where $F=2P_2+P_3-Q$.  Then $P_3\geq Q$, since the maximum eigenvalue of $Q$ is at most 1, and so $F\geq 0$.  Let $\ket{\psi}=\ket{0}\ket{u}+\ket{1}\ket{v}$ be a normalized eigenvector of~$\tau^{\star\Gamma}$ with a negative coresponding eigenvalue and consider the projector $P_4=\ket{\psi}\bra{\psi}$. We may write $F$ and $P_4$ in matrix form as
\[
  F=\begin{pmatrix}
    A & B\\
    B^\dagger & C
   \end{pmatrix}
   \hspace{5mm}\text{and}\hspace{5mm}
    P_4=\begin{pmatrix}
    \ket{u}\bra{u} & \ket{u}\bra{v}\\
    \ket{v}\bra{u} & \ket{v}\bra{v}
   \end{pmatrix}.
\]
Since $P_4$ is in the negative eigenspace of $\tau^{\star\Gamma}$, we have $P_2\geq P_4$ and so $F\geq 2P_4$.  Hence $A\geq  2\ket{u}\bra{u}$ and $C\geq  2\ket{v}\bra{v}$, and thus the largest eigenvalues of $A$ and~$C$ are at least as large as $2\langle u|u\rangle$ and $2\langle v|v\rangle$ respectively.  As in the previous case, $\langle u|u\rangle+\langle v|v\rangle=1$ and so the largest eigenvalue among all eigenvalues of $A$ and $C$ is at least~1.  Since
\[
 F^\Gamma=\begin{pmatrix}
    A & B^\dagger\\
    B & C
   \end{pmatrix},
\]
by the Couchy interlacing theorem, the largest eigenvalue of $F^\Gamma$ is at least 1, and thus $\phi=\mathds{1}-F^\Gamma$ cannot be strictly positive.
\end{proof}

\subsection{Comparing the Rains bound to the logarithmic negativity}

It is also possible to compare the Rains bound to the logarithmic negativity. From the definition, the Rains bound of a state~$\rho$ is not larger than $LN(\rho)$. The Rains bound of a state~$\rho$ is equal to its logarithmic negativity $LN(\rho)$ if and only if the matrix $\tau^\star$ that minimizes the Rains bound is proportional to~$\rho$, namely $\tau^\star=\frac{1}{\Tr\abs{\rho^\Gamma}}\rho$.  Indeed, if $\tau^\star=\frac{1}{\Tr\abs{\rho^\Gamma}}\rho$ minimizes the Rains bound for $\rho$, then
\begin{multline*}
 R(\rho)=S(\rho\Vert\tau^\star)=S\left(\rho\,\middle\Vert\,\tfrac{1}{\Tr|\rho^\Gamma|}\rho\right)=\\=S(\rho\Vert\rho)+\log\Tr|\rho^\Gamma|=LN(\rho),
\end{multline*}
since $S(\rho\| \rho)=0$.

 From Theorem~\ref{thm:RainsBound}, we have that $R(\rho)=LN(\rho)$ if and only if \mbox{$\phi=L_{\tau^\star}(\rho)$} defines a supporting functional such that
$ \Tr[\phi\tau]\leq \Tr[\phi\tau^\star]$ for all $\tau\in\mathcal{T}$.  
For any positive constant $c\in(0,+\infty)$ and matrix \mbox{$A\in\H_{n}$}, 
the derivative operator $L_{cA}$ is equal to~$\frac{1}{c}L_{A}$. 
With $c=\frac{1}{\Tr\abs{\rho^\Gamma}}$, we have that $\tau^\star=c\rho$ and thus
\[
 \phi=\Tr\abs{\rho^\Gamma}P_\rho.
\]
Hence the Rains bound for a state $\rho$ is equal to its logarithmic negativity if and only if
\begin{equation}
\label{eq:compareRainsSupp}
 \Tr[P_\rho\tau]\leq \Tr[P_\rho\rho]=1
\end{equation}
for all $\tau\in\mathcal{T}$.

In particular, this means that the Rains bound is strictly smaller than the logarithmic negativity for any full-rank non-PPT state. Indeed, if \mbox{$R(\rho)=LN(\rho)$} for some full-rank $\rho\in\H_{n,+,1}$ then the matrix \mbox{$\phi=\Tr\abs{\rho^\Gamma}\mathds{1}$} defines a supporting functional of $\mathcal{T}$, since $P_\rho=\mathds{1}$. However, $\Tr[\tau]$ obtains its maximal value over $\mathcal{T}$ if and only if $\tau\in\mathcal{P}$, and thus $\rho$ must have been PPT to begin with. Thus, if $R(\rho)=LN(\rho)$ for some $\rho\not\in\mathcal{P}$ then $\rho$ has at least one zero eigenvalue.

\section{Other applications}
\label{sec:OtherApplications}
Instead of using $g(x)=\log(x)$ and only considering $\rho$ to be an arbitrary quantum state as we did for the REE and the Rains bound, we may use other, more general concave functions~$g$ and positive matrices~$\rho$ and analyze when a matrix~$\sigma^\star$ minimizes the function \mbox{$f_\rho(\sigma)=-\Tr[\rho g(\sigma)]$}.  
Since the criterion for optimality in Theorem~\ref{thm:maintheorem} only supposes an arbitrary convex function~\mbox{$f:\H_{n}\rightarrow\RR^{+\infty}$}, we may use this analyis on  to produce similar hyperplane criteria. In the following, we examine this criterion as it applies to other quantities of interest in quantum information, such as $h_{\rm SEP}$ and entropy-like quantities, e.g. the relative R\'enyi entropies and arbitrary quasi-entropies.

\subsection{On \texorpdfstring{$h_{\textrm{SEP}}(M)$}{hSEP} and similar quantities}

By setting $\rho=M$, where $0\leq M\leq\mathds{1}$, using the identity function $g(x)=x$, and optimizing over the set of separable states, we arrive at $f(\sigma)=\Tr[M\sigma]$. Maximizing this function over $\mathcal{D}$, we define
\[
 h_{\rm SEP}(M)=\max_{\sigma\in\mathcal{D}} \Tr[M\sigma].
\]
This quantity is related to the problem of finding the maximum output $\infty$-norm of a quantum channel, which is important for proving additivity and multiplicativity for random channels~\cite{Harrow2013}. Any quantum channel from an $n_1$-dimensional quantum system to an \mbox{$n_2$-dimensional} quantum system can be written as $\mathcal{E}(\rho)=\Tr_{env}[V\rho V^\dagger]$ for some isometry $V:\CC^{n_1}\rightarrow \CC^{n_2}\otimes\CC^{n_{env}}$~\cite{Choi1975}. The quantum channel can be identified with the operator $M=VV^\dagger$. The maximal output $p$-norm of the channel $\mathcal{E}$ is defined by
\[
 \norm{\mathcal{E}}_{1\rightarrow p}=\max_{\rho} \norm{\mathcal{E}(\rho)}_p
\]
where $\norm{A}_p=(\Tr\abs{A}^p)^{1/p}$ is the Schatten $p$-norm. It turns out that the maximal output $\infty$-norm can be given by $ \norm{\mathcal{E}}_{1\rightarrow \infty}=h_{\rm SEP}(M)$~\cite{Harrow2013}.  In addition, $h_{\rm SEP}$ has a natural interpretation in terms of determining the maximum probabilities of success in QMA(2) protocols~\cite{Montanaro2013}.

Whereas numerically calculating this quantity within an accuracy of $1/\textrm{poly}(n)$ is known to be an NP-hard problem~\cite{Gurvits2003}, it might be useful to study the converse problem. That is, given a state $\sigma^\star$ on the boundary of~$\mathcal{D}$, characterizing all matrices $0\leq\phi\leq\mathds{1}$ that define supporting functionals of the form
\[
 \Tr[\phi \sigma]\leq \Tr[\phi\sigma^\star] \,\,\textrm{ for all }\sigma\in\mathcal{D}
\]
allows us to characterize all of the matrices $M$ for which~$\sigma^\star$ maximizes~$\Tr[M\sigma^\star]$.

Since the set of PPT states is much easier characterize than the set of separable states, maximizing~$\Tr[M\sigma]$ over $\mathcal{P}$ instead gives an approachable upper bound to~$h_{\rm SEP}(M)$. That is, analyzing
\[
 h_{\rm PPT}(M)=\max_{\sigma\in\mathcal{P}} \Tr[M\sigma]
\]
where $\mathcal{P}=\rm{PPT}$ is the set of states with positive partial transpose. Using this as an upper bound of $\norm{\mathcal{E}}_{1\rightarrow\infty}$ still gives meaningful useful results~\cite{Harrow2013}. Furthermore, the supporting functionals of $\mathcal{P}$ that are maximized by a state $\sigma^\star$ on the boundary are easy to find (see eq.~\eqref{eq:PPThyperplane} in section~\ref{sec:CSScriterion}).

\subsection{Quasi \texorpdfstring{$f$}{f}-relative entropies}
Many important properties of the relative entropy (such as the convexity and monotonicity) are only due to the concavity of the function $g(x)=\log(x)$, so at first glance there is nothing special about this choice in terms analyzing divergence of two quantum states.  It is important to understand more general entropy-like functions for quantum states in order to glean a better understanding of generalized pseudo-distance measures on the space of quantum states. Many such functions have already been introduced and analyzed~\cite{Petz1986,Ohya1993,Petz2010,Sharma2011,Hiai2011}. 

The most general form of the so-called quasi $f$-relative entropies first appeared in~\cite{Petz1986}. A more simplified form that we analyze here, supposing that $\rho$ is a strictly positive matrix, is given by~\cite{Sharma2011}
\begin{equation}
 S_f(\rho\| \sigma)=\sum_{i}p_i\left\langle\psi_i\abs{f\left(\frac{\sigma}{p_i}\right)} \psi_i\right\rangle
\end{equation}
where $\rho=\sum_i p_i\ket{\psi_i}\bra{\psi_i}$ is the spectral decomposition of~$\rho$.  As long as the function $f:(0,\infty)\rightarrow\RR$ is operator convex, the $f$-relative entropy satisfies some very important criteria that make it a useful quantity to study. For example, it has been shown~\cite{Petz1986} that~$S_f(\rho\| \sigma)$ is jointly convex and satisfies monotonicity,~i.e.~$S_f(\Lambda(\rho)\| \Lambda(\sigma))\leq S_f(\rho\| \sigma)$ for any completely positive trace-preserving map~$\Lambda$.
The quasi-entropies are generally not additive or subadditive, however, so their physical significance is limited.

Although it may be possible to extend the definitions of the quasi-entropies to singular matrices, for the time being we may assume that $f:(0,\infty)\rightarrow\RR$ is is well-defined and that both $\sigma$ and $\rho$ are strictly positive. Consider a convex subset $\mathcal{C}\subset\H_{n,+}$, and define quantities analogous to the relative entropy of entanglement $E_\mathcal{D}$ or $E_\mathcal{P}$ in the following manner. Define the \emph{$f$-relative entropy with respect to~$\mathcal{C}$} as 
\[
 E_\mathcal{C}^f(\rho)=\min_{\sigma\in\mathcal{C}} S_f(\rho\| \sigma).
\]
We may use our analysis from section~\ref{sec:maintheorem} to determine necessary and sufficient conditions for when a matrix~$\sigma^\star\in\mathcal{C}$ optimizes the $f$-relative entropy for a matrix $\rho$, i.e. when $E_\mathcal{C}^f(\rho)=S_f(\rho\| \sigma^\star)$.  According to Theorem~\ref{thm:maintheorem}, this occurs if and only if, for all $\sigma\in\mathcal{C}$, we have that $\left.\frac{d}{dt}S_f(\rho\| \sigma^\star+t(\sigma-\sigma^\star))\right|_{t=0^+}\geq0$.  

We can evaluate this derivative by making use of the Fr\'echet derivative of the function $f$ and find
\begin{multline*}
 \left.\frac{d}{dt}S_f(\rho\Vert\sigma^\star+t(\sigma-\sigma^\star))\right|_{t=0^+} = \\
 =\sum_i\left\langle \psi_i\left|  D_{f,\frac{\sigma^\star}{p_i}}\left(\sigma-\sigma^\star\right)\right| \psi_i \right\rangle.
\end{multline*}
Since the Fr\'echet derivative operator $D_{f,\frac{\sigma^\star}{p_i}}$ is self-adjoint, we can rewrite the right-hand side of the above equation as 
\[
\sum_i\Tr\left[D_{f,\frac{\sigma^\star}{p_i}}\left(\ket{\psi_i}\bra{\psi_i}\right)\left(\sigma-\sigma^\star\right)\right].
\]
 Thus, a matrix $\sigma^\star\in\mathcal{C}$ minimizes the $f$-relative entropy with respect to a matrix $\rho$ if and only if the matrix
\begin{equation}
 \label{eq:quasihyperplane}
 \phi=-\sum_i D_{f,\frac{\sigma^\star}{p_i}}\left(\ket{\psi_i}\bra{\psi_i}\right)
\end{equation}
defines a supporting functional of~$\mathcal{C}$ of the form
\[
 \Tr[\phi\sigma]\leq\Tr[\phi\sigma^\star] \,\,\textrm{ for all }\sigma\in\mathcal{C}.
\]
This characterization of the supporting functionals of a convex set might yield interesting information about the original $S_f(\rho\| \sigma)$ quantity, as our analysis of $S(\rho\| \sigma)$ has shown.

Indeed, with the choice $f(x)=-\log(x)$ the standard definition of the relative entropy is recovered. Furthermore, the desired supporting functionals in eq.~\eqref{eq:quasihyperplane} reduce to
\begin{align*}
\phi=-\sum_i D_{f,\frac{\sigma^\star}{p_i}}\left(\ket{\psi_i}\bra{\psi_i}\right)&=\sum_iL_{\frac{\sigma^\star}{p_i}}\left(  \ket{\psi_i}\bra{\psi_i}\right) \\  &= L_{\sigma^\star}\left(\sum_i p_i \ket{\psi_i}\bra{\psi_i}\right) \\&=  L_{\sigma^\star}(\rho),
\end{align*}
which is exactly the form of the hyperplanes found in the analysis of the relative entropy in section \ref{sec:REE}.

Other standard choices of $f$ yield additional well-known entropy-like quantites. For example, the choice $f_\alpha(x)=x^\alpha$ produces a quantity that is related to the relative R\'enyi entropy, which we study in the following section.

\subsection{Relative R\'enyi entropy}

Choosing the function $f_\alpha(x)=x^\alpha$ for $\alpha$ in the range~$(0,1)\cup(1,\infty)$, the $f_\alpha$-relative entropy becomes
\[
 S_{f_\alpha}(\rho\| \sigma)=\Tr[\rho^{\alpha-1}\sigma^\alpha].
\]
This is closely related to the standard definition of the $\alpha$-relative R\'enyi entropy~\cite{Ohya1993}
\begin{equation}
\tilde D_\alpha(\rho\| \sigma)=\frac{1}{\alpha-1}\log\Tr[\rho^\alpha\sigma^{1-\alpha}],
\end{equation}
which is equivalent to $\frac{1}{\alpha-1}\log S_{f_{1-\alpha}}(\rho\| \sigma)$.  The $\alpha$-relative R\'enyi entropy is jointly convex for $\alpha\in(0,1)$~\cite{Mosonyi2009} and satisfies the data-processing inequality for $\alpha\geq \tfrac{1}{2}$~\cite{Frank2013}.

For $\alpha<1$, note that minimizing $D_\alpha(\rho\| \sigma)$ for a fixed~$\rho$ is equivalent to minimizing $\Tr[\rho^\alpha\sigma^{1-\alpha}]$, so that the condition that a matrix $\sigma^\star\in\mathcal{C}$ minimizes the $\alpha$-relative R\'enyi entropy with respect to~$\rho$ over~$\mathcal{C}$ is given by
\[
 \frac{d}{dt}\Tr\left[\rho^\alpha\big(\sigma^\star+t(\sigma-\sigma^\star)\big)^{1-\alpha}\right]\bigg|_{t=0^+}\geq0 \,\,\,\textrm{ for all }\sigma\in\mathcal{C},
\]
which we can analze by determining the Fr\'echet derivative of functions of the type $x^\alpha$.

Considering the function $g_\alpha(x)=x^\alpha$, given a matrix $A\in\H_n$ with strictly positive eigenvalues $\{a_1,\dots, a_n\}$, we have
\[
 \left[T_{f_\alpha,A}\right]_{ij}=\left\{\begin{array}{ll}
          \dfrac{a_i^\alpha-a_j^\alpha}{a_i-a_j} & a_i\neq a_j\\
          \alpha a_i^{\alpha-1}& a_i=a_j
         \end{array}
\right. 
\]
such that the Fr\'echet derivative is given by 
\[D_{f_\alpha,A}(B)=\left.\frac{d}{dt}(A+tB)^\alpha\right|_{t=1}=T_{f_\alpha,A}\circ B.
 \]
Note that $T_{f_\alpha,A}\circ \mathds{1}=\alpha A^{\alpha-1}$. 

For the $\alpha$-relative R\'enyi entropies, we have
\[
 \left.\frac{d}{dt}\big(\sigma^\star+t(\sigma-\sigma^\star)\big)^{1-\alpha}\right|_{t=0^+}= D_{f_{1-\alpha},\sigma^\star}(\sigma-\sigma^\star)
\]
and so a matrix $\sigma^\star$ optimizes $\tilde D(\rho\| \sigma)$ over~$\mathcal{C}$ if and only~if the matrix
\[
 \phi=-D_{f_{1-\alpha},\sigma^\star}(\rho^\alpha)
\]
defines a supporting functional of~$\mathcal{C}$ that is maximized by $\sigma^\star$ of the form $\Tr[\phi\sigma]\leq\Tr[\phi\sigma^\star]$ for all $\sigma\in\mathcal{C}$.  

Therefore, a matrix $\sigma^\star$ is optimal for a state $\rho$ in this case if and only if there is a matrix $\phi$ defining a supporting functional of the desired form such that
\[
 \rho=\left(-D_{f_{1-\alpha},\sigma^\star}^{-1}(\phi)\right)^{1/\alpha}.
\]

\subsection{Sandwiched relative R\'enyi entropy}
A generalization of the relative R\'enyi entropy that was recently proposed is another quantity that we can study. For $\alpha\in(0,1)\cup(1,\infty)$ and matrices $\rho,\sigma\in\H_{n,+}$, the \emph{order $\alpha$ quantum R\'enyi divergence} (or also called the \emph{``sandwiched'' $\alpha$-relative R\'enyi entropy}) is defined as~\cite{Muller-Lennert2013}:
\begin{equation}
 D_\alpha(\rho\| \sigma)=\frac{1}{\alpha-1}\log\left(\Tr\left[\left(\sigma^\frac{1-\alpha}{2\alpha}\rho\sigma^\frac{1-\alpha}{2\alpha}\right)^\alpha\right]\right)
\end{equation}
and reduces to the standard $\alpha$-relative R\'enyi entropy $\tilde D_\alpha(\rho\| \sigma)$ when $\rho$ and $\sigma$ commute. 
This quantity has been shown to be jointly convex~\cite{Frank2013} when $\alpha\in[\tfrac{1}{2},1)$ and when the argument~$\rho$ is restricted to matrices with unit trace, and it also satasifies the data processing inequality for $\alpha\geq\tfrac{1}{2}$~\cite{Beigi2013}. In the limit $\alpha\rightarrow1$, it reduces to the standard quantum relative entropy $S(\rho\| \sigma)$. For $\alpha=\tfrac{1}{2}$, the quantity $D_{1/2}(\rho\| \sigma)=-2\log\!\norm{\!\sqrt{\rho}\sqrt{\sigma}}$ is closely related to the quantum fidelity~\cite{Datta2013a}. It is also positive $D_\alpha(\rho\| \sigma)\geq0$ for positive matrices~$\rho$ and~$\sigma$ and vanishes if and only if~$\rho=\sigma$.

As in the previous examples, we can use the conditions in Theorem~\ref{thm:maintheorem} to determine when a matrix $\sigma^\star$ minimizes the R\'enyi divergence of $\rho$ over a convex set~$\mathcal{C}$.  This occurs when
\[
 \left.\frac{d}{dt}D_\alpha\big(\rho\big\| \sigma^\star+t(\sigma-\sigma^\star)\big)\right|_{t=0^+}\geq0\,\,\,\textrm{ for all }\sigma\in\mathcal{C}.
\]
Analogous to the many of the other cases studied here, this can be turned into a supporting functional criterion of the form
\[
\Tr[\phi\sigma]\leq\Tr[\phi\sigma^\star]
\]
for all $\sigma\in\mathcal{C}$. Here, $\phi$ is the matrix 
\begin{equation}
 \label{eq:alphasandwichedhyperplane}
 \phi=- D_{f_\beta,\sigma^\star}\left(\left\{
 \sigma^{\star{{-\beta}}},\left(\sigma^{\star\beta}\rho\sigma^{\star\beta}\right)^\alpha\right\}
 \right)
\end{equation}
where $\{A,B\}=AB+BA$ is the anti-commutator and $\beta=\tfrac{1-\alpha}{2\alpha}$. Thus, the R\'enyi divergence of order $\alpha\in[\tfrac{1}{2},1)$ of a matrix $\rho\in\H_{n,+}$ with respect to~$\mathcal{C}$ is minimized by $\sigma^\star$ on the boundary of~$\mathcal{C}$, 
\[
 \min_{\sigma\in\mathcal{C}} D_\alpha(\rho\| \sigma)=D_\alpha(\rho\| \sigma^\star),
\]
if any only if the matrix $\phi$ in eq.~\eqref{eq:alphasandwichedhyperplane} defines a supporting functional of~$\mathcal{C}$ at $\sigma^\star$.

\section{Conclusion}
\label{sec:conclusion}

In conclusion, for a convex function $f:\H_n\rightarrow\RR^{+\infty}$ and a convex subset $\mathcal{C}\subset\H_n$, we present a criterion to solve the converse convex optimization problem of determining when a matrix $\sigma^\star\in\mathcal{C}$ minmizes $f$ over $\mathcal{C}$, i.e. such that $f(\sigma^\star)=\displaystyle\min_{\sigma\in\mathcal{C}}f(\sigma)$. This criterion can usually be given in terms of a supporting functional of $\mathcal{C}$ at~$\sigma^\star$. Given a concave analytic function $g:(0,\infty)\rightarrow\RR$, this approach allows us to determine all matrices $\rho$ such that~$\sigma^\star$ minimizes the function $f_\rho=-\Tr[\rho g(\sigma)]$ over $\mathcal{C}$ by characterizing the supporting hyperplanes of $\mathcal{C}$ at $\sigma^\star$. In particular, given a matrix~$\sigma^\star$, we use this analysis to produce closed fomulae for the relative entropy of entanglement and the Rains bound for \emph{all} states for which $\sigma^\star$ minimizes this quantity, which hold regardless of the dimensionality of the system.  Moreover, this allows us to show that the Rains bound and the relative entropy of entanglement coincide for all states for which at least one subsystem is a qubit.

We also find supporting functional criteria to determine when a state $\sigma^\star$ minimizes other various important quantites in quantum information, such as, for example, the generalized relative entropies. 

\hspace{10mm}

\emph{Acknowledgements} -- The authors are grateful for useful discussions with Yuriy Zinchenko for assistance in producing numerical examples. M.G. and G.G. are supported by NSERC and \href{http://www.pims.math.ca/scientific/collaborative-research-groups/mathematics-quantum-information-2010-2013}{PIMS CRG MQI}. S.F. is supported by NSF grant DMS-1216393.


\bibliographystyle{my_apsrev_bib}
\bibliography{/home/mark/Documents/library.bib}

\end{document}